\documentclass[Royal,times]{sagej_pp}
\usepackage{moreverb,url}

\usepackage[colorlinks,bookmarksopen,bookmarksnumbered,citecolor=red,urlcolor=red]{hyperref}

\newcommand\BibTeX{{\rmfamily B\kern-.05em \textsc{i\kern-.025em b}\kern-.08em
T\kern-.1667em\lower.7ex\hbox{E}\kern-.125emX}}

\newtheorem{theorem}{Theorem}[section]

\newtheorem{proposition}[theorem]{Proposition}

\theoremstyle{definition}
\newtheorem{definition}[theorem]{Definition}

\theoremstyle{remark}
\newtheorem{remark}[theorem]{Remark}
\begin{document}

\runninghead{Information Entropy of the Financial Market}

\title{Information Entropy of the Financial Market: Modelling Random Processes Using Open Quantum Systems}

\author{Will Hicks\affilnum{1}}

\affiliation{\affilnum{1} Memorial University of Newfoundland}

\corrauth{Will Hicks}

\email{whicks7940@googlemail.com}

\begin{abstract}
We discuss the role of information entropy on the behaviour of random processes, and how this might take effect in the dynamics of financial market prices. We then go on to show how the Open Quantum Systems approach can be used as a more flexible alternative to classical methods in terms of modelling the entropy gain of a random process. We start by describing an open quantum system that can be used to model the state of a financial market. We then go on to show how to represent an essentially classical diffusion in this framework. Finally, we show how by relaxing certain assumptions, one can generate interesting and essentially non-classical results, which are highlighted through numerical simulations.
\end{abstract}
\maketitle
\section{Introduction}
Despite the fact that probabilistic methods are widely applied in finance, one can argue that changes in market prices are not random. Real factors (low profits impacting an equity price, poor economic performance impacting FX rates etc) lie behind price changes. The reason why we turn to the study of random variables, and the framework of probability, is to do with our lack of {\em information} regarding what the future price will be. For this reason the consideration of the information entropy is an important factor.

The open quantum systems approach outlined in this article presents a means by which one can study the impact of information entropy on the time evolution of a system using the framework of quantum probability. Furthermore, we show how the seemingly random evolution of a traded price can arise from simple interactions between the market and its' environment.

We start in section \ref{entropy} by discussing the ways in which quantum probability models differ from models based on commutative assumptions, and classical probability. In section \ref{GenSetup} we propose a Hilbert space representation of the market environment, that allows us to maximise the benefits of the open quantum systems approach. We then derive the general form for the master equation in section \ref{LME}. In particular we show how one can apply ladder operators, such as those discussed in \cite{Bagarello}, in deriving the master equation in an open systems framework. The Markovian approximation is given in section \ref{Markov}. Then in section \ref{Gaussian}, we go on to derive a Gaussian process whereby the state diffuses in an essentially classical fashion. That is, a finite dimensional quantum state that is diagonalized relative to the traded price operator remains diagonalized into the future.

Finally in section \ref{NG}, we show how the methodology can easily be extended to cover non-classical domains, where we show that statistical properties such as variance and kurtosis are linked to the degree of information entropy as we look further into the future.

In this article, the focus is primarily on outlining the general methodology, and discussing the financial interpretation of the general steps we are taking. We defer the task of deriving specific models that can be directly applied, to a future study. It should also be noted that whilst the general master equation we derive in section \ref{LME} can equally apply where we use an infinite dimensional Hilbert space for the financial market (eg $\mathcal{H}_{mkt}=L^2(\mathbb{R})$), we focus in this article on representing the market in finite dimensions ($\mathcal{H}_{mkt}=\mathbb{C}^N$).
\section{Entropy of the Financial Market:}\label{entropy}
\subsection{Entropy in the Classical Case:}
For discrete probabilities $p_i:i=1\dots N$, the Shannon entropy is given by (see for example \cite{NC} chapter 11):
\begin{align}\label{shannon}
H(\{p_i\},i=1\dots N)&=-\sum_{i=1}^Np_i\log(p_i)
\end{align}
If $x_i$ labels the outcome that occurs with probability $p_i$, then equation \ref{shannon} can be interpreted as a measure of the information that would be gained if we were to find out the future outcome for certain, rather than simply knowing the probability. So clearly:
\begin{itemize}
\item If we know we will get the outcome $x_k$ with probability $p_k=1$, then the entropy is zero.
\item If we have no information whatsoever, then the probability for each outcome is the same: $p_i=1/N$, for $i=1\dots N$. In this case the entropy is maximized:
\begin{align*}
H(\{p_i\},i=1\dots N)&=\log(N)
\end{align*}
\end{itemize}
Now, assume that we have a classical discrete (eg integer valued) random walk, where the probability distribution for each step is independent and identically distributed. If we label the probability distribution for the position after $n$ steps as: $P_n$, then it follows from \cite{ABBN} Theorem 1 that:
\begin{align}\label{ent_mono}
H(P_{n+1})>H(P_n)
\end{align}
In other words, after each step, the entropy monotonically increases.

When one thinks about random walks, used to represent financial variables, one tends to think about concepts such as the variance, rather than entropy. For example, we may wish to track the variance for a discrete observable acting on $\mathbb{C}^N$ after $n$ steps of a random walk.

If we now label the probability of finding the outcome $x_i$, after $n$ steps of the random walk, as $P_n^i$, then we have:
\begin{align*}
Var(X,P_n)&=E^{P_n}[X^2]-(E^{P_n}[X]^2)\\
E^{P_n}[f(X)]&=\sum_{i=1}^Nf(x_i)P_n^i
\end{align*}
Assuming that the classical random walk makes independent and identically distributed steps, in addition to equation \ref{ent_mono}, we have:
\begin{align}\label{var_mono}
Var(X,P_{n+1})>Var(X,P_n)
\end{align}
%Next 3 paragraphs are key. This is where we justify why entropy is important. I think it needs more work...
When analysing financial market time-series, and the prices of listed option contracts, one tends to consider equation \ref{var_mono}, rather than equation \ref{ent_mono}. Furthermore, once one has decided on the probability distribution we wish to use for future financial market returns, the entropy is fixed. In the quantum case, which we describe in section \ref{ent_quant}, this simple relationship does not apply. One can have 2 different market states, that have the same probability distribution, but different levels of entropy.

In section \ref{ent_ex}, we describe a highly simplified example of a financial situation which describes just such a phenomenon. That is where the 2 different financial markets have the same probability distribution for a traded market price, but differing amounts of the Von-Neumann entropy, corresponding to differing amounts of financial information.

Then in section \ref{NG} we show examples of where the amount of information we have regarding the market state, as measured by the Von-Neumann entropy, impacts its' evolution into the future.
\subsection{Entropy in the Quantum Case:}\label{ent_quant}
In the more general case, the state of the market is described by a quantum state acting on a Hilbert space. For example, if we assume there are $N$ different possible outcomes: $x_i$, $i=1\dots N$, with probabilities: $p_i$, then we would set the Hilbert space as:
\begin{align*}
\mathcal{H}_{mkt}&=\mathbb{C}^N
\end{align*}
The possible outcomes are encoded in an operator acting on the Hilbert space:
\begin{align}\label{X_comm}
X&=\sum_{i=1}^Nx_i|e_i\rangle\langle e_i|
\end{align}
Then the market, and associated probabilities, are determined by a density matrix: $\rho$ acting on $\mathbb{C}^N$, so that we have:
\begin{align*}
\rho&=\sum_{i,j=1}^N\rho_{ij}|e_i\rangle\langle e_j|\text{, }\sum_{i=1}^N\rho_{ii}=1\\
p_i&=Tr[\rho\mathcal{P}_i]\text{, }\mathcal{P}_i=|e_i\rangle\langle e_i|\\
&=\rho_{ii}\\
E^{\rho}[X]&=Tr[X\rho]\\
&=\sum_{i=1}^Np_ix_i
\end{align*}
In the quantum case, the Von Neumann entropy is given by:
\begin{align}\label{VN ent}
H(\rho)&=-Tr[\rho\log\rho]
\end{align}
Note that in the event that we have: $\rho=\sum_{i=1}^Np_i|e_i\rangle\langle e_i|$, then the Von-Neumann entropy (\ref{VN ent}) and the Shannon entropy (\ref{shannon}) coincide. For this reason, we describe this as a {\em classical} state, and use the notation:
\begin{align*}
\rho_{classical}&=\sum_{i=1}^Np_i|e_i\rangle\langle e_i|
\end{align*}
To understand the importance of the quantum approach, first consider the following:
\begin{proposition}\label{entropy_proposition}
Let the operator $X$ be given by equation  \ref{X_comm}, and define the set of projection operators:
\begin{align}\label{mini_def_f(X)_entropy}
\mathcal{P}_i&=|e_i\rangle\langle e_i|
\end{align}
Finally, consider the set of density matrices: $\mathcal{A}$ for which we have, for $\rho\in\mathcal{A}$:
\begin{align*}
E[\mathcal{P}_i]&=Tr[\rho\mathcal{P}_i]\\
&=p_i
\end{align*}
In other words $\mathcal{A}$ is the set of density matrices which fixes the probability of finding the price $x_i$, for each $i=1$ to $N$. Then the classical density matrix:
\begin{align*}
\rho_{classical}&=\sum_{i=1}^Np_i|e_i\rangle\langle e_i|
\end{align*}
Maximises the Von-Neumann entropy within $\mathcal{A}$.
\end{proposition}
\begin{proof}
See appendix \ref{appendix}.
\end{proof}
%Thus, from proposition \ref{entropy_proposition}, if we have a market where we will see price $x_i$ with probability $p_i$, then $\rho_{classical}$ is the state about which we have the least information. 
%In section \ref{ent_ex}, we present a simple example, using the case of 2 different contracts referencing the same underlying assets (listed stock price vs listed option price), where one can have different levels of information without impacting the probabilities we use. Then in section  \ref{NG} we present examples of where this can impact the future evolution of statistical properties.
\subsection{Entropy Example: Listed Stock Price, vs Listed Option Price}\label{ent_ex}
Proposition \ref{entropy_proposition} shows that given a specific finite dimensional probability distribution for a traded financial market price, the classical state represents the state about which we have the least information. Alternatively, it represents the case where the most information is gained from finding out what the price will be with certainty, having previously only known the probability distribution.

In many circumstances, there may be uncertainty regarding a particular traded price, but where additional information {\em is} available to the market. For example the specific market mechanisms that go into determining a trade execution price, or the official end of day close price. Alternatively the size \& motivation of market participants. The quantum probability framework discussed in this chapter enables a way to distinguish between situations where the probability law for the price is the same, but the overall information available to the market is different. With a view to illustrating the point, we consider the following toy example:
\begin{itemize}
\item Due to imperfections in market price fixing mechanism (for example non-zero bid-offer spread). There are 3 possible prices for the traded price of an asset. The market Hilbert space is therefore set to: $\mathcal{H}=\mathbb{C}^3$.
\item The 3 possible prices are $x_1$, $x_2$, $x_3$, associated to the eigenvectors $|e_i\rangle$, $i=1,2,3$.
\item We have an operator that acts on the market state, returning the trade price for an asset:
\begin{align*}
X=\sum_{i=1}^3x_i|e_i\rangle\langle e_i|
\end{align*}
\item We also consider the traded price operator: $O$, for a Strangle option consisting of an `at the money' listed call option and an `at the money' listed put option.
\item Since the listed put and call options are both `at the money', the option has the lowest value $o_-$ if the market is in the middle eigenstate: $|e_2\rangle$.
\item We assume the option has the value $o_{1,+}$ in the eigenstate: $|v_1\rangle=\frac{|e_1\rangle+|e_3\rangle}{\sqrt{2}}$, and the value $o_{2,+}$ in the eigenstate $|v_2\rangle=\frac{|e_1\rangle-|e_3\rangle}{\sqrt{2}}$.
\item Note that $|e_1\rangle$, $|v_1\rangle$, and $|v_2\rangle$  are an alternative orthonormal basis and we can write:
\begin{align*}
O=o_-|e_2\rangle\langle e_2|+o_{1,+}|v_1\rangle\langle v_1|+o_{2,+}|v_2\rangle\langle v_2|
\end{align*}
\end{itemize}
We first consider the case that the market state is given by:
\begin{align}\label{toy_classical_state}
\rho_{classical}&=0.25|e_1\rangle\langle e_1|+0.5|e_2\rangle\langle e_2|+0.25|e_3\rangle\langle e_3|
\end{align}
This has a 25\% chance of finding the $X$ price of $x_1$, a 50\% chance of finding the $X$ price of $x_2$, and a 25\% chance of finding the $X$ price of $x_3$. Similarly we find a 25\% chance each of finding the $O$ price of $o_{1,+}$ or $o_{2,+}$, and a 50\% chance of finding the $O$ price of $o_-$. The Von-Neumann entropy is given in this case by:
\begin{align*}
H(\rho_{classical})\approx 1.04
\end{align*}
Next consider the case:
\begin{align}\label{toy_q_state}
\rho_{quantum}&=0.25|e_1\rangle\langle e_1|+0.5|e_2\rangle\langle e_2|+0.25|e_3\rangle\langle e_3|+0.25|e_1\rangle\langle e_3|+0.25|e_3\rangle\langle e_1|
\end{align}
This has a lower value for the Von-Neumann entropy:
\begin{align*}
H(\rho_{quantum})\approx 0.69
\end{align*}
even though the discrete probability distribution for the traded price: $X$ is the same. In this case the lower Von-Neumann entropy reflects the fact we have additional information regarding the market price of the Strangle option $O$, that doesn't effect the probability distribution for $X$, whereas, for the state \ref{toy_classical_state}, we have no more information about the traded Option price $O$, than we do about the traded stock price $X$.

In the state \ref{toy_q_state}, we can eliminate the possibility of finding the traded price $o_{2,+}$. In fact we have:
\begin{align*}
\rho_{quantum}&=0.5|v_1\rangle\langle v_1|+0.5|e_2\rangle\langle e_2|
\end{align*}
meaning that there is a 50\% chance of finding the value $o_{1,+}$, and a 50\% chance of finding the value $o_-$. For example, if $o_{2,+}>o_{1,+}$, this could reflect the possibility that investors will not pay more than $o_{1,+}$ for this Strangle option. This additional information is reflected in the lower entropy.
\subsection{Remark on Diagonalization in a Financial Context}
In general, when considering operators on a finite dimensional Hilbert space, the density matrix that represents the state of the system can only be spoken about as diagonalized relative to a particular orthonormal basis, or alternatively relative to a particular self-adjoint operator. The key difference in the case whereby the Hilbert space represents a particular market is that one generally identifies a primary traded instrument. For example for traded equities, this could be the listed stock, or for equity indices, the front month futures contract etc.

In the example above, we illustrate the case, where the listed option contract and listed stock have a different eigenbasis. In this case, this reflects the fact that even though we may know the stock will return the price: $x_1$, further information is required to identify the price for the listed option, which could still return the price $o_{1,+}$ or $o_{2,+}$, depending on the demand for options. Note that this is achieved {\em without using any additional parameters}, such as one describing volatility.

Alternatively, if we use the Hilbert space to represent the market on the day of the option expiry, this difference could represent the case where to trade a listed stock, one would match a market bid or offer price, whilst the value of the listed option was determined by the official end of day close price.

There are various ways this could be extended. For example using different operators to represent different order types. However, we defer further discussion for a future study (see for example \cite{WH_MIS}).
\section{Setting up the General Framework:}\label{GenSetup}
\subsection{Defining the Market Hilbert Space:}
In this section we define the Hilbert space representation for the financial market. We discuss the financial interpretation behind the setup in section \ref{section_on_interp} before defining some of the key operators we will be using in section \ref{KeyOp}. The general form for the Lindblad Master Equation is derived in section \ref{LME}.

We follow the basic approach outlined in \cite{BP} section 3. That is, the full system is represented by the tensor product of the market Hilbert space (labelled $\mathcal{H}_{mkt}$), with the external environment Hilbert space (labelled $\mathcal{H}_{env}$):
\begin{align}\label{H_sys}
\mathcal{H}&=\mathcal{H}_{mkt}\otimes\mathcal{H}_{env}
\end{align}
For reasons that we discuss in section \ref{section_on_interp}, we let the environment Hilbert space be given by:
\begin{align}\label{H_env}
\mathcal{H}_{env}&=\mathbb{C}^K\otimes L^2[\mathcal{K}]\text{, }K\geq 2
\end{align}
Where, $\mathcal{K}$ is a bounded subset of $\mathbb{R}$. For example: $\mathcal{K}=[-L,L]$ for some $L>0$. The full system Hamiltonian is given by:
\begin{align}\label{Hamiltonian}
H&=H_I+(\mathbb{I}\otimes H_{env})\\
H_I&=\sqrt{\kappa\gamma}\sum_{\alpha\in \{u,d\}}A_{\alpha}\otimes B_{\alpha}\nonumber
\end{align}
In equation \ref{Hamiltonian}, $H_I$ models the interaction between the financial market and its environment. $A_u$ and $A_d$ act on the market Hilbert space: $\mathcal{H}_{mkt}$,  $\gamma$, $\kappa$ are constants. $B_u$ and $B_d$ act on the environment space and are defined by:
\begin{align}\label{B_def}
B_u&=\sum_{i=1}^{K-1}|e_{i+1}\rangle\langle e_i|\otimes\mathbb{I}\text{, }B_d=\sum_{i=1}^{K-1}|e_i\rangle\langle e_{i+1}|\otimes\mathbb{I}
\end{align}
$H_{env}$ is the environment Hamiltonian, which we assume has the form:
\begin{align}\label{product}
H_{env}&=\gamma\sum_{l=1}^Kl|e_l\rangle\langle e_l|\otimes H'
\end{align}
Where $H'$ acts on the space: $L^2[\mathcal{K}]$.
\subsection{Financial Interpretation of the Hilbert Space Structure:}\label{section_on_interp}
The operators we are interested in, act on the market Hilbert space: $\mathcal{H}_{mkt}$. For now, we take this Hilbert space to represent the various potential buyers \& sellers that make up the market for a particular tradeable security. For example, it might represent a stock exchange if the traded security was a listed equity price. If we used: $\mathcal{H}_{mkt}=\mathbb{C}^N$, we could write:
\begin{align*}
X&=\sum_{i=1}^Nx_i|e_i\rangle\langle e_i|\otimes\mathbb{I}\otimes\mathbb{I}
\end{align*}
where the values $x_i$ represent the price for a financial asset in the event that the market is found in the state: $|e_i\rangle\langle e_i|$.
 
The space $\mathcal{H}_{env}$ represents the general environment in which the trading activity occurs. There are two components of this space. The first is a finite dimensional Hilbert space $\mathbb{C}^K$. We interpret the $K$ different eigenstates for this space as $K$ different levels of market risk appetite. For example the eigenstate $|e_1\rangle\langle e_1|$ would represent the most {\em bearish} state for the market, whereby participants are looking to reduce risk exposure, with a view to protecting the values of their investment portfolios. Similarly $|e_K\rangle\langle e_K|$ would represent the most {\em bullish} state for the market, with investors looking to build up their risk exposure with a view to maximising returns on their investment portfolios.

The operator $A_u\otimes B_u$ increases the level of market bullishness. The operator $B_u$ shifts the background environment to a higher level of risk appetite, and the operator $A_u$ controls the resulting impact on the operators we measure (ie the market price). The operator $A_d\otimes B_d$ has the opposite effect: shifting the market to a lower level of risk appetite, and tracking the resulting price impact. In this case, we have assumed that an increase in risk appetite leads to an increase in the market price (and vice versa).

Since we have no way of measuring the state of the environment when we take a measurement of the price (eg by executing a trade), we gain no information regarding the state acting on $\mathcal{H}_{env}$. This operation is carried out using the partial trace:
\begin{align*}
E^{\rho}[X]&=Tr[\rho_{mkt}(t)X]\\
\rho_{mkt}(t)&=Tr_{env}[\rho(t)]
\end{align*}
Finally, we consider the second component of the environment space: $L^2[\mathcal{K}]$. Partly the introduction of a space with a continuous spectrum is pragmatic. When we go on to discuss the time evolution, we will need to calculate expressions like:
\begin{align*}
f_{ud}(t,s)&=Tr[B_u(t)B_d(s)\rho^I_{env}(s)]
\end{align*}
As shown in \cite{RivasHuelga}, these will not generally converge unless one integrates over a continuous spectrum, and it will certainly not be possible to apply the strong coupling limit. From a financial perspective, this space, and the operator $H'$, ensure that the system Hamiltonian returns a continuous energy spectrum, and will control the energy gained/lost when the environment shifts to a higher/lower level of risk appetite.
\subsection{Choice of the Market Hamiltonian}
For any classical model for the dynamics of a traded asset price, for example one based on Ito calculus, there are 2 components of the time evolution:
\begin{itemize}
\item The deterministic component of the time evolution, usually termed the {\em drift} component.
\item The random component of the time evolution. This is often modelled classically using a Wiener process, and is often termed the {\em diffusion} component.
\end{itemize}
The classical approach to risk neutral pricing pricing requires the definition of Martingale measure. I.e, a probability measure $Q$ such that for a derivative payout (with value at time $t$ of $V_t$) referencing a traded underlying asset (with price at time $t$ of $X_t$), we have:
\begin{align*}
E^Q[V_t]=V_0
\end{align*}
The deterministic component of the classical time evolution is determined by the choice of the Martingale measure. If we were to set: $\mathcal{H}_{mkt}=L^2[\mathbb{R}]$, we could represent the deterministic drift at the risk free rate: $r$, using a market Hamiltonian given by:
\begin{align}\label{market_ham}
H_{mkt}=-ir\frac{\partial}{\partial x}
\end{align}
Where $r$ is a constant that represents the risk free interest rate. Applying \ref{market_ham}, we have:
\begin{align*}
|\psi(x,t)\rangle&=e^{iH_{mkt}t}|\psi(x)\rangle\\
&=|\psi(x-rt)\rangle
\end{align*}
Note that since:
\begin{align*}
[A_u,H_{mkt}]&=[A_d,H_{mkt}]\\
&=0
\end{align*}
we find that equation \ref{key_prop_res}, is translation invariant under this Hamiltonian choice. In other words, using $r\neq 0$ will not impact the dynamics beyond the translation: $|\psi(x)\rangle\rightarrow|\psi(x-rt)\rangle$, and we are free to choose $r=0$ without loss of generality. From a financial perspective, setting $r=0$ means we are modelling the dynamics of forward prices, rather than the current market price (also called the spot price). For many traded underlyings, the market liquidity for forward contracts is sufficient such that it is standard practice to hedge using forward prices, with a maturity that matches the maturity of the derivative, rather than using spot prices.

In theory, we could include kinetic energy and potential energy terms in the market Hamiltonian:
\begin{align}\label{market_ham_2}
H_{mkt}&=-ir\frac{\partial}{\partial x}-\frac{1}{2m}\frac{\partial^2}{\partial x^2}+V(x)
\end{align}
This is an important avenue for research, and is discussed further in \cite{Bagarello2}, and \cite{Hicks4}. After setting $r=0$, and hence $H_{mkt}=0$ in equation \ref{market_ham}, the dynamics of the market are driven by the interaction with the environment space, which provides the random noise element. By applying \ref{market_ham_2}, we are assuming that the market state has its' own internal energy, which will also drive the dynamics. This represents a non-deterministic component to the time evolution, which has no classical counterpart. However, the focus for the current research project is to apply quantum models for the random component of the classical models for the dynamics of traded asset prices. Therefore, for the time being, we choose: $H_{mkt}=\mathbb{I}$, thus ensuring that the only time evolution in the operators we are interested in, comes from the interaction with the environment.
\subsection{Defining Key Operators:}\label{KeyOp}
\begin{definition}
We assume $\rho_B$ is a stationary state, and in particular assume that $[H',\rho_B]=0$. Then the quantum state at time $t$, acting on the Hilbert space \ref{H_sys}, with the environment space being given by \ref{H_env}, can be represented as the following sum:
\begin{align}\label{state_ass}
\rho(t)&=\sum_{l,m=1}^K\rho_{mkt}^{lm}(t)\otimes |e_l\rangle\langle e_m|\otimes\rho_B
\end{align}
\end{definition}
The Lindblad master equation, which we derive in section \ref{LME}, will determine the time evolution of the state: \ref{state_ass}, in the Schr{\"o}dinger interpretation. However, when deriving this equation, we will also make use of the following definition for operators in the interaction picture:
\begin{definition}\label{def_int_pic}
Let $A$ be an operator on the Hilbert space $\mathcal{H}$ given by \ref{H_sys}, with the system Hamiltonian given by \ref{Hamiltonian}. We also let the environment Hamiltonian: $H_{env}$, be defined by equation \ref{product}. Then we define the interaction picture operators as follows:
\begin{align*}
A^I(t)&=e^{i(\mathbb{I}\otimes H_{env})t}Ae^{-i(\mathbb{I}\otimes H_{env})t}
\end{align*}
\end{definition}
\begin{proposition}\label{prop_on_int_pic}
The interaction picture state for \ref{state_ass} is given by:
\begin{align*}
\rho^I(t)&=\sum_{l,m=1}^K\rho^{lm}_{mkt}(t)\otimes |e_l\rangle\langle e_m|\otimes e^{i\gamma(l-m)tH'}\rho_B
\end{align*}
\end{proposition}
\begin{proof}
We have:
\begin{align}\label{prove_int_prop}
\rho^I(t)&=e^{i(\mathbb{I}\otimes H_{env})t}\big(\rho(t)\big)e^{-i(\mathbb{I}\otimes\ H_{env})t}\nonumber\\
&=\sum_{l,m=1}^K\rho_{mkt}^{lm}\otimes\big(e^{iH_{env}t}\big[|e_l\rangle\langle e_m|\otimes\rho_B\big]e^{-iH_{env}t}\big)
\end{align}
Using equation \ref{product} for the environment Hamiltonian, we get:
\begin{align*}
e^{iH_{env}t}&=\exp\Big(i\gamma t\sum_{l=1}^Kl|e_l\rangle\langle e_l|\otimes H'\Big)\\
&=\sum_{l=1}^K|e_l\rangle\langle e_l|\otimes e^{i\gamma ltH'}
\end{align*}
Applying this to \ref{prove_int_prop}, we get (since $[H',\rho_B]=0$):
\begin{align*}
\rho^I(t)&=\sum_{l,m=1}^K\rho_{mkt}^{lm}(t)\otimes |e_l\rangle\langle e_m|\otimes e^{i\gamma(l-m)tH'}\rho_B
\end{align*}
\end{proof}
\begin{proposition}\label{B(t)_prop}
The interaction picture for $H_I$ is given by:
\begin{align*}
H_I(t)&=\sqrt{\kappa\gamma}(A_u\otimes B_u(t)+A_d\otimes B_d(t))
\end{align*}
Where we have:
\begin{align}\label{B(t)}
B_u(t)&=\sum_{l=1}^{K-1}|e_{l+1}\rangle\langle e_l|\otimes e^{i\gamma H' t}\\
B_d(t)&=\sum_{l=1}^{K-1}|e_{l}\rangle\langle e_{l+1}|\otimes e^{-i\gamma H' t}\nonumber
\end{align}
\end{proposition}
\begin{proof}
The system Hamiltonian (excluding the interaction Hamiltonian), combining equations \ref{H_sys} and \ref{H_env} is:
\begin{align*}
H_{sys}&=\mathbb{I}\otimes\sum_{l=1}^Kl|e_l\rangle\langle e_l|\otimes H'
\end{align*}
Therefore, for $A_u$, and $A_d$ we get:
\begin{align*}
A_u(t)&=e^{iH_{sys}t}A_ue^{-iH_{sys}t}\\
&=e^{i\mathbb{I}t}A_ue^{-i\mathbb{I}t}=A_u\\
A_d(t)&=e^{iH_{sys}t}A_de^{-iH_{sys}t}\\
&=e^{i\mathbb{I}t}A_de^{-i\mathbb{I}t}=A_d
\end{align*}
For the operators $B_u$, and $B_d$, defined in equation \ref{B_def} we note that, as after equation \ref{prove_int_prop},
\begin{align*}
e^{iH_{env}t}&=\sum_{l=1}^K|e_l\rangle\langle e_l|\otimes e^{il\gamma H't}
\end{align*}
Applying this to \eqref{B_def}, we get:
\begin{align*}
B_u(t)&=\sum_{l=1}^{K-1}|e_{l+1}\rangle\langle e_l|\otimes e^{i(l+1)\gamma H't}e^{-il\gamma H't}\\
&=\sum_{l=1}^{K-1}|e_{l+1}\rangle\langle e_l|\otimes e^{i\gamma H't}
\end{align*}
Similarly, we get:
\begin{align*}
B_d(t)&=\sum_{l=1}^{K-1}|e_l\rangle\langle e_{l+1}|\otimes e^{il\gamma H't} e^{-i(l+1)\gamma H't}\\
&=\sum_{l=1}^{K-1}|e_l\rangle\langle e_{l+1}|\otimes e^{-i\gamma H't}
\end{align*}
\end{proof}
The next result concerns taking the trace over the Hilbert space $L^2[\mathcal{K}]$.
\begin{proposition}\label{FuncAnal}
We assume that the Hamiltonian $H'$ in equation \ref{product} is self-adjoint, and that we have (for some orthonormal basis): $\rho_B=\sum_{i=1}^{\infty}q_i|e_i\rangle\langle e_i|$. Then $H'$ has the spectral resolution:
\begin{align}\label{spectral_res}
H'&=\int_{\mathbb{R}}\omega P(d\omega)
\end{align}
$P$ is a projection valued measure in the sense of \cite{Hall} definition 7.10. That is, $P$ maps Borel subsets of $\mathbb{R}$ to projection operators acting on the Hilbert space: $L^2[\mathcal{K}]$. Then we have:
\begin{align}\label{bath_state}
Tr[f(H')\rho_B]&=\int_{\mathbb{R}}f(\omega)d\mu^{(H',\rho_B)}(\omega)
\end{align}
Where $\mu^{(H',\rho_B)}$ is a probability measure on $\mathbb{R}$. Furthermore, since $H'$ is a bounded operator we have:
\begin{align*}
\mu^{(H',\rho_B)}(E)&<\infty\text{, for all }E\subset\mathbb{R}
\end{align*}
\end{proposition}
\begin{proof}
The existence of the spectral resolution \ref{spectral_res} and the projection valued measure $P$ follows from the assumption that $H'$ is bounded and self adjoint, and from \cite{Hall} Theorem 7.12. From \cite{Hall} definition 7.13, it follows that we can define:
\begin{align*}
f(H')&=\int_{\sigma(H')}f(\omega)P(d\omega)
\end{align*}
Where $\sigma(H')$ is the spectrum of $H'$. Since $L^2[\mathcal{K}]$ is separable, we can write:
\begin{align*}
Tr[f(H')\rho_B]&=\sum_{i=1}^{\infty}q_i\langle e_i|\int_{\sigma(H')}f(\omega)P(d\omega)|e_i\rangle
\end{align*}
For some $q_i$, and orthonormal basis vectors: $|e_i\rangle$
\begin{align*}
&=\int_{\sigma(H')}f(\omega)\sum_{i=1}^{\infty}q_i\langle e_i|P(d\omega)|e_i\rangle\\
&=\int_{\sigma(H')}f(\omega)d\mu^{(H'\rho_B)}(\omega)\\
\mu^{(H',\rho_B)}(E)&=\sum_{i=1}^{\infty}q_i\langle e_i|P(E)|e_i\rangle\text{, for }E\subset\mathbb{R}
\end{align*}
We have that $\langle\psi|P(E)|\psi\rangle<\infty$ for all $E\subset\mathbb{R}$, and all $\psi\in L^2[\mathcal{K}]$, and we can extend the integral to $\mathbb{R}$ by defining $P(E)=0$ for $E\cup\sigma(H')=\emptyset$.
\end{proof}
Finally, we can obtain a Markov approximation for the financial market dynamics, by applying the following.
\begin{proposition}[Strong Coupling Limit]\label{SCL}
Assume $H'$ is bounded and self adjoint, and is represented by the spectral resolution: \ref{spectral_res}. Furthermore, we assume that the measure: $\mu^{(H',\rho_B)}$, given by proposition \ref{FuncAnal}, is absolutely continuous, and that:
\begin{align*}
d\mu^{(H'\rho_B)}(\omega)&=p(\omega)d\omega
\end{align*}
For some density function $p(\omega)$. Then we have, as $\gamma\rightarrow\infty$:
\begin{align}
\kappa\gamma Tr[B_{\alpha}(t)B_{\beta}(s)\rho^I_{env}]&\rightarrow\kappa Tr[B_{\alpha}B_{\beta}\rho^I_{env}]\delta(t-s)\label{f_res}\\
\kappa\gamma Tr[B_{\alpha}(s)B_{\beta}(t)\rho^I_{env}]&\rightarrow\kappa Tr[B_{\alpha}B_{\beta}\rho^I_{env}]\delta(t-s)\label{g_res}\\
\alpha,\beta &\in\{u,d\}\nonumber
\end{align}
Here $\rho^I_{env}(s)$ represents an interaction picture state acting on the environment Hilbert space, with the Hamiltonian \ref{product}, and the stationary state $\rho_B$ acting on $L^2[\mathcal{K}]$.
\end{proposition}
\begin{proof}
The general state acting on $\mathcal{H}_{env}$ (with $\rho_B$ acting on $L^2[\mathcal{K}]$) can be written:
\begin{align*}
\rho_{env}&=\sum_{l,m=1}^Kr_{lm}|e_l\rangle\langle e_m|\otimes \rho_B
\end{align*}
First note that from the proof of \ref{prop_on_int_pic} it follows that under the Hamiltonian \ref{product}, the interaction picture state is given by:
\begin{align*}
\rho^I_{env}(s)&=\sum_{l,m=1}^Kr_{lm}|e_l\rangle\langle e_m|\otimes e^{i\gamma(l-m)\omega sH'}\rho_B
\end{align*}
Inserting from proposition \ref{B(t)_prop}, we get, where $S(\alpha)=+1$ for $\alpha=u$, and $S(\alpha)=-1$ for $\alpha=d$:
\begin{align}\label{interim_res}
\kappa\gamma Tr[B_{\alpha}(t)B_{\beta}(s)\rho^I_{env}(s)]&=\kappa\gamma Tr\Big[B_{\alpha}B_{\beta}\sum_{l,m=1}^Kr_{lm}|e_l\rangle\langle e_m|\Big]\\
\times &Tr\Big[\Big(e^{i\gamma(l-m)sH'}e^{i\gamma S(\alpha)tH'}e^{i\gamma S(\beta)sH'}\Big)\rho_B\Big]\nonumber
\end{align}
For $\alpha=\beta=u$, \ref{interim_res} becomes:
\begin{align}\label{uu}
\kappa\gamma Tr[B_u(t)B_u(s)\rho^I_{env}(s)]&=\kappa\sum_{l=3}^Kr_{(l-2)l}\Big(\gamma Tr\Big[e^{-2i\gamma sH'}e^{i\gamma sH'}e^{i\gamma tH'}\rho_B\Big]\Big)\nonumber\\
&=\kappa\sum_{l=1}^{K-2}r_{l(l+2)}\Big(\gamma Tr\Big[e^{i\gamma(t-s)H'}\rho_B\Big]\Big)
\end{align}
Where the second line follows since only terms where $l-m=-2$ contribute to the trace, and $S(u)=1$.
For $\alpha=\beta=d$, we have:
\begin{align}\label{dd}
\kappa\gamma Tr[B_d(t)B_d(s)\rho^I_{env}(s)]&=\kappa\sum_{l=1}^{K-2}r_{(l+2)l}\Big(\gamma Tr\Big[e^{2i\gamma sH'}e^{-i\gamma sH'}e^{-i\gamma tH'}\rho_B\Big]\Big)\nonumber\\
&=\kappa\sum_{l=1}^{K-2}r_{(l+2)l}\Big(\gamma Tr\Big[e^{i\gamma (s-t)H'}\rho_B\Big]\Big)
\end{align}
Finally, for $\alpha=u,\beta=d$ and $\alpha=d,\beta=d$:
\begin{align}\label{ud,du}
\kappa\gamma Tr[B_u(t)B_d(s)\rho^I_{env}(s)]&=\kappa\sum_{l=1}^{K-1}r_{ll}\Big(\gamma Tr\Big[e^{i\gamma(t-s)H'}\rho_B\Big]\Big)\\
\kappa\gamma Tr[B_d(t)B_u(s)\rho^I_{env}(s)]&=\kappa\sum_{l=1}^{K-1}r_{ll}\Big(\gamma Tr\Big[e^{i\gamma(s-t)H'}\rho_B\Big]\Big)\nonumber
\end{align}
From proposition \ref{FuncAnal} we have:
\begin{align*}
\gamma Tr\big[e^{i\gamma(t-s)H'}\rho_B\big]&=\int_{\mathbb{R}}e^{i\gamma(t-s)\omega}d\mu^{(H'\rho_B)}(\omega)\\
&=\gamma\int_{\mathbb{R}}e^{i\gamma(t-s)\omega}p(\omega)d\omega
\end{align*}
Where the second line follows from the assumption that $\mu^{(H',\rho_B)}$ is absolutely continuous, and can therefore be written using a probability density function $p(\omega)$. Taking the limit $\gamma\rightarrow\infty$ we get:
\begin{align}\label{FT}
\gamma Tr\big[e^{i\gamma(t-s)H'}\rho_B\big]&=\gamma\int_{\mathbb{R}}e^{i\gamma\omega(t-s)}p(\omega)d\omega\nonumber\\
&=\int_{\mathbb{R}}e^{i(t-s)u}p(u/\gamma)du\nonumber\\
&\rightarrow\int_{-\infty}^{\infty}e^{i(t-s)u}du\text{, as }\gamma\rightarrow\infty\nonumber\\
&=\delta(t-s)
\end{align}
The result then follows by applying this to equations: \ref{uu}, \ref{dd} and \ref{ud,du}.
\end{proof}
\section{Time Evolution Mechanism:}\label{LME}
In the next proposition we derive the general form for the time evolution, before showing how the strong coupling limit (proposition \ref{SCL}) leads to Markovian dynamics in section \ref{Markov}.
\begin{proposition}\label{key_proposition}
We let the environment Hilbert space be given by \ref{H_env}, the full system Hamiltonian by: \ref{Hamiltonian}, with $B_u$, and $B_d$ given by \ref{B_def}, and where $H_{env}$ has the product form \ref{product}. Finally, we assume the state starts with the form given by \ref{state_ass}, and that $\rho_B$ remains in a stationary state (Born approximation), and is therefore independent of time. Then the dynamics of the reduced density matrix are given by:
\begin{align}\label{key_prop_res}
\frac{d\rho_{mkt}(t)}{dt}&=-i[H_I(t),\rho^I(0)]\\
&-\int_0^tds\bigg(\sum_{l,m=1}^K\sum_{\alpha,\beta\in\{u,d\}}f^{lm}_{\alpha\beta}(t,s)\Big(A_{\alpha}A_{\beta}\rho^{lm}_{mkt}(s)-A_{\beta}\rho^{lm}_{mkt}(s)A_{\alpha}\Big)\nonumber\\
&+g^{lm}_{\alpha\beta}(t,s)\Big(\rho^{lm}_{mkt}(s)A_{\alpha}A_{\beta}-A_{\alpha}\rho^{lm}_{mkt}(s)A_{\beta}\Big)\bigg)\nonumber
\end{align}
Where we denote:
\begin{align*}
f^{lm}_{\alpha\beta}(t,s)&=\kappa\gamma Tr[B_{\alpha}(t)B_{\beta}(s)r_{lm}(s)|e_l\rangle\langle e_m|\otimes e^{i\gamma(l-m)sH'}\rho_B]\text{, }\alpha\text{, }\beta\in \{u,d\}\\
g^{lm}_{\alpha\beta}(t,s)&=\kappa\gamma Tr[B_{\alpha}(s)B_{\beta}(t)r_{lm}(s)|e_l\rangle\langle e_m|\otimes e^{i\gamma(l-m)sH'}\rho_B]\text{, }\alpha\text{, }\beta\in \{u,d\}
\end{align*}
\end{proposition}
\begin{proof}
We work in the interaction picture, as defined in definition \ref{def_int_pic}. From proposition \ref{prop_on_int_pic}, we have:
\begin{align*}
\rho^I(t)&=\sum_{l,m=1}^K\rho^{lm}_{mkt}\otimes r_{lm}(t)|e_l\rangle\langle e_m|\otimes e^{i\gamma(l-m)tH'}\rho_B
\end{align*}
Also, from proposition \ref{B(t)_prop} we have:
\begin{align*}
H_I(t)&=\sqrt{\kappa\gamma}(A_u\otimes B_u(t)+A_d\otimes B_d(t))\\
B_u(t)&=\sum_{l=1}^{K-1}|e_{l+1}\rangle\langle e_l|\otimes e^{i\gamma H' t}\\
B_d(t)&=\sum_{l=1}^{K-1}|e_{l}\rangle\langle e_{l+1}|\otimes e^{-i\gamma H' t}
\end{align*}
Next we feed the interaction picture state: $\rho^I(t)$ into the Von Neumann equation to get:
\begin{align}\label{VN_int}
\frac{\partial}{\partial t}\rho^I(t)&=-i[H_I(t),\rho^I(t)]\nonumber\\
\rho^I(t)&=\rho^I(0)-i\int^t_0ds[H_I(s),\rho^I(s)]
\end{align}
Note that if we define:
\begin{align*}
\rho^I_{mkt}(t)&=Tr_{env}[\rho^I(t)]
\end{align*}
Then, from definition \ref{def_int_pic}, we have:
\begin{align*}
\rho^I_{mkt}(t)&=Tr_{env}\Big[e^{i(\mathbb{I}\otimes H_{env})t)}\Big(\sum_{l,m=1}^K\rho^{lm}_{mkt}(t)\otimes r_{lm}|e_l\rangle\langle e_m|\otimes\rho_B\Big)e^{-i(\mathbb{I}\otimes H_{env})t)}\Big]\\
&=\sum_{l,m=1}^K\rho^{lm}_{mkt}(t)Tr\Big[e^{iH_{env}t}\Big(r_{lm}|e_l\rangle\langle e_m|\otimes\rho_B\Big)e^{-iH_{env}t)}\Big]\\
&=\sum_{l,m=1}^K\rho^{lm}_{mkt}(t)Tr[r_{lm}|e_l\rangle\langle e_m|\otimes\rho_B]\\
&=Tr_{env}[\rho(t)]\\
&=\rho_{mkt}(t)
\end{align*}
Now, inserting \ref{VN_int} back into the interaction picture Von-Neumann equation, before taking the partial trace over the environment, gives:
\begin{align}\label{VN_eqn}
\frac{\partial}{\partial t}\rho_{mkt}(t)=-iTr_{env}[H_I(t),\rho^I(0)]&-\int_0^tds Tr_{env}\Big[H_I(t),[H_I(s),\rho^I(s)]\Big]\\
=-iTr_{env}[H_I(t),\rho^I(0)]-&\int_0^tdsTr_{env}\Big[H_I(t),H_I(s)\rho^I(s)-\rho^I(s)H_I(s)\Big]\nonumber\\
=-iTr_{env}[H_I(t),\rho^I(0)]&-\int_0^tdsTr_{env}\Big(H_I(t)H_I(s)\rho^I(s)\nonumber\\
-H_I(t)\rho^I(s)H_I(s)&-H_I(s)\rho^I(s)H_I(t)+\rho^I(s)H_I(s)H_I(t)\Big)\nonumber
\end{align}
Note also that by the cyclicity of the trace we have (where $\rho^I_{env}(s)$ is a state acting on the Hilbert space $\mathcal{H}_{env}=\mathbb{C}^K\otimes L^2[\mathcal{K}]$):
\begin{align*}
Tr[B_u(t)B_d(s)\rho^I_{env}(s)]&=Tr[B_d(s)\rho^I_{env}(s)B_u(t)]=Tr[\rho^I_{env}(s)B_u(t)B_d(s)]
\end{align*}
In the $\rho^I(s)H_I(s)H_I(t)$ term in \ref{VN_eqn}, we can therefore write:
\begin{align}\label{rhoHH}
\sum_{l,m=1}^K\rho^{lm}_{mkt}(s)A_uA_dTr[\rho^{I,lm}_{env}(s)B_u(s)B_d(t)]
&=\sum_{l,m=1}^K\rho^{lm}_{mkt}(s)A_uA_dTr[B_u(s)B_d(t)\rho^{I,lm}_{env}(s)]\nonumber\\
\sum_{l,m=1}^K\rho^{lm}_{mkt}(s)A_dA_uTr[\rho^{I,lm}_{env}(s)B_d(s)B_u(t)]
&=\sum_{l,m=1}^K\rho^{lm}_{mkt}(s)A_dA_uTr[B_d(s)B_u(t)\rho^{I,lm}_{env}(s)]\nonumber\\
\sum_{l,m=1}^K\rho^{lm}_{mkt}(s)A_uA_uTr[\rho^{I,lm}_{env}(s)B_u(s)B_u(t)]
&=\sum_{l,m=1}^K\rho^{lm}_{mkt}(s)A_uA_uTr[B_u(s)B_u(t)\rho^{I,lm}_{env}(s)]\nonumber\\
\sum_{l,m=1}^K\rho^{lm}_{mkt}(s)A_dA_dTr[\rho^{I,lm}_{env}(s)B_d(s)B_d(t)]
&=\sum_{l,m=1}^K\rho^{lm}_{mkt}(s)A_dA_dTr[B_d(s)B_d(t)\rho^{I,lm}_{env}(s)]\nonumber\\
\rho^{I,lm}(s)&=r_{lm}(s)|e_l\rangle\langle e_m|\otimes\rho_B
\end{align}
We write:
\begin{align}\label{f,g notation}
f^{lm}_{\alpha\beta}(t,s)=\kappa\gamma Tr[B_{\alpha}(t)B_{\beta}(s)\rho^{I,lm}_{env}(s)]\\
g^{lm}_{\alpha\beta}(t,s)=\kappa\gamma Tr[B_{\alpha}(s)B_{\beta}(t)\rho^{I,lm}_{env}(s)]\nonumber
\end{align}
Therefore \ref{rhoHH} becomes:
\begin{align}\label{firstOQS}
\sum_{l,m=1}^K\rho^{lm}_{mkt}(s)A_uA_dg^{lm}_{ud}(t,s)+
\sum_{l,m=1}^K\rho^{lm}_{mkt}(s)A_dA_ug^{lm}_{du}(t,s)\nonumber\\
+\sum_{l,m=1}^K\rho^{lm}_{mkt}(s)A_uA_ug^{lm}_{uu}(t,s)
+\sum_{l,m=1}^K\rho^{lm}_{mkt}(s)A_dA_dg^{lm}_{dd}(t,s)
\end{align}
Then in the $H_I(t)\rho^I(s)H_I(s)$ term we get:
\begin{align}\label{secondOQS}
\sum_{l,m=1}^KA_u\rho^{lm}_{mkt}(s)A_dg^{lm}_{ud}(t,s)+\sum_{l,m=1}^KA_d\rho^{lm}_{mkt}(s)A_ug^{lm}_{du}(t,s)\nonumber\\
\sum_{l,m=1}^KA_u\rho^{lm}_{mkt}(s)A_ug^{lm}_{uu}(t,s)
+\sum_{l,m=1}^KA_d\rho^{lm}_{mkt}(s)A_dg^{lm}_{dd}(t,s)
\end{align}
Then in the $H_I(s)\rho^I(s)H_I(t)$ term:
\begin{align}\label{third}
\sum_{l,m=1}^KA_u\rho^{lm}_{mkt}(s)A_df^{lm}_{ud}(t,s)+\sum_{l,m=1}^KA_d\rho^{lm}_{mkt}(s)A_uf^{lm}_{du}(t,s)\nonumber\\
+\sum_{l,m=1}^KA_u\rho^{lm}_{mkt}(s)A_uf^{lm}_{uu}(t,s)
+\sum_{l,m=1}^KA_d\rho^{lm}_{mkt}(s)A_df^{lm}_{dd}(t,s)
\end{align}
and finally in the $H_I(t)H_I(s)\rho^I(s)$ term:
\begin{align}\label{fourth}
\sum_{l,m=1}^KA_uA_d\rho^{lm}_{mkt}(s)f^{lm}_{ud}(t,s)+\sum_{l,m=1}^KA_dA_u\rho^{lm}_{mkt}(s)f^{lm}_{du}(t,s)\nonumber\\
\sum_{l,m=1}^KA_uA_u\rho^{lm}_{mkt}(s)f^{lm}_{uu}(t,s)
+\sum_{l,m=1}^KA_dA_d\rho^{lm}_{mkt}(s)f^{lm}_{dd}(t,s)
\end{align}
Collecting together \ref{firstOQS}, \ref{secondOQS}, \ref{third}, and \ref{fourth}, leads to \ref{key_prop_res} as required.
\end{proof}
\subsection{Markovian Approximation:}\label{Markov}
The majority of models of the financial market, applied by practioners, assume Markovian dymamics, partly as a result of tractability, and partly due to the fact that it can be shown that non-Markovian models are (according to many reasonable definitions) arbitrageable. The discussion of whether non-Markovian models of the financial market are reasonable, and investigations into their properties, is an active area of research (see for example \cite{BjorkWick}, \cite{SottVal} and \cite{Oksendal2}). However, this is not the focus of the current research, and we would therefore like to apply a Markovian approximation to proposition \ref{key_proposition}. With this in mind, in this section, we indicate how this can be achieved, by sketching out the mechanisms by which near Markovian dynamics can arise. We use the Strong Coupling Limit, outlined for example in \cite{BP}, section 3.3, and \cite{RivasHuelga} section 6. We then go on to apply the Markovian approximation for the remainder of the chapter.

In order to derive the Markovian approximation for proposition \ref{key_proposition}, we assume that the environment remains in a time independent maximum entropy state given by:
\begin{align}\label{identity}
\rho_{env}&=\frac{1}{K}\sum_{i=1}^K|e_i\rangle\langle e_i|\otimes\rho_B
\end{align}
Under \ref{identity} we find that $[H_{env},\rho_{env}]=0$, and thus:
\begin{align*}
\rho^I_{env}&=e^{iH_{env}t}\rho_{env}e^{-iH_{env}t}\\
&=\rho_{env}
\end{align*}
We now apply propositions \ref{B(t)_prop} and \ref{SCL} to proposition \ref{key_proposition} to derive the Markovian dynamics.
\begin{proposition}[Born-Markov Approximation]\label{BMA}
After applying the strong coupling limit \ref{SCL}, and assuming the environment state is given by equation \ref{identity}, proposition \ref{key_proposition} becomes:
\begin{align}\label{Markov Res}
\frac{d\rho_{mkt}(t)}{dt}&=-Tr_{env}[H_I(t),\rho^I(0)]\\
&+\sigma^2\Big(A_u\rho_{mkt}(t)A_d+A_d\rho_{mkt}(t)A_u-\frac{1}{2}\{A_uA_d+A_dA_u,\rho_{mkt}(t)\}\Big)\nonumber
\end{align}
Where we denote:
\begin{align}\label{vol_vs_K}
\sigma^2&=\frac{\kappa(K-1)}{K}
\end{align}
\end{proposition}
\begin{proof}
First note that, given equation \ref{identity}, we have:
\begin{align*}
f^{ll}_{ud}(t,s)&=f^{ll}_{du}(t,s)=g^{ll}_{ud}(t,s)=g^{ll}_{du}(t,s)\\
&=\frac{\kappa\delta(t-s)}{K}\text{, }l=\{2,\dots,K-1\}\\
f^{KK}_{ud}(t,s)&=g^{KK}_{ud}(t,s)=f^{11}_{du}(t,s)=g^{11}_{du}(t,s)\\
&=\frac{\kappa\delta(t-s)}{K}
\end{align*}
With all other terms: $f^{lm}_{\alpha\beta}(t,s)$, $g^{lm}_{\alpha\beta}(t,s)$ equal to zero. The result then follows by feeding this into equation \ref{key_prop_res}, and integrating from $0$ to $t$.
\end{proof}
\begin{remark}[Alternative Approach: The Weak Coupling Limit]
Before moving on, we briefly consider the main alternative means by which open quantum systems can be approximated using Markovian dynamics. That is the weak coupling limit (see for example \cite{BP}, \cite{RivasHuelga}). In the weak coupling limit, one assumes that the interaction between the system and the environment is weak in comparison to the market Hamiltonian \& the environment Hamiltonian. This means that the typical variation time for the market space: $\tau_{mkt}$, becomes very large as the strength of the interaction with the environment becomes smaller, as does the ratio $\tau_{mkt}/\tau_{env}$, where $\tau_{env}$ is the typical evolution time for the environment. In other words, the random change, that arises from the interaction Hamiltonian: $H_I$, causes only slow evolution of the market price. We have not pursued this approach for two reasons:
\begin{itemize}
\item Firstly, it requires the expansion of the $A_u$ and $A_d$ operators in terms of the eigenvectors for the market Hamiltonian: $H_{mkt}$. We would like the model to work in the event that we have $H_{mkt}=0$.
\item Secondly, the assumption that the interaction with the environment leads to a slow evolution is not consistent with the fractal nature of the market price. That is, plotting a time-series of price movements over arbitrarily small time intervals (eg price changes every few seconds) leads to a qualitatively similar result as plotting the time-series showing price changes over longer time intervals (eg daily price changes). See for example \cite{Mandelbrot} for further discussion.
\end{itemize}
In the strong coupling limit, we assume that the environment Hamiltonian scales with a constant $\alpha$, and the interaction Hamiltonian by $\sqrt{\alpha}$. This ensures that the typical variation time from the environment: $\tau_{env}$ becomes very small. In the limit of $\alpha\rightarrow\infty$, the environment settles back to equilibrium effectively instantaneously, and the market price undergoes Markovian evolution, where the memory of prior price movements is essentially forgotten and has no impact of future price changes.
\end{remark}
\subsection{Example: Gaussian Case}\label{Gaussian}
In this section we assume: $\mathcal{H}_{mkt}=\mathbb{C}^N$, and seek to derive classical diffusion dynamics from the approach in section \ref{Markov}, before discussing how non-classical diffusion can occur in section \ref{NG}. 

First assume that $\rho_{mkt}(0)$ is a diagonal matrix, so that:
\begin{align}\label{rho_0_class_fin}
\rho_{mkt}(0)&=\sum_{i=1}^Np_i(0)|e_i\rangle\langle e_i|
\end{align}
Note, that the state $\ref{rho_0_class_fin}$ can be considered a classical state, as we now explain. The price operator $X$ in this case given by:
\begin{align*}
X&=\sum_{i=1}^Nx_i|e_i\rangle\langle e_i|
\end{align*}
Where $x_i\in\mathbb{R}$ is the real valued price eigenvalue that is returned for the eigenstate $|e_i\rangle$, and that $x_i>x_j$ for $i>j$. This leads to:
\begin{align*}
E[X]&=Tr[X\rho_{mkt}(0)]\\
&=\sum_{i=1}^Np_i(0)x_i
\end{align*}
So $p_i(0)$ represents the probability of finding the traded market price $x_i$ (at time $t=0$). In this setup, $A_u$ represents the operator that shifts the market price higher by one notch, and $A_d$ shifts the market price down by one notch:
\begin{definition}\label{A_finite}
\begin{align*}
A_u&=\sum_{i=1}^{N-1}|e_{i+1}\rangle\langle e_i|\text{, }A_d=\sum_{i=1}^{N-1}|e_i\rangle\langle e_{i+1}|
\end{align*}
\end{definition}
The following proposition shows that under equation \ref{key_prop_res}, with $A_u$, $A_d$ given by definition \ref{A_finite}, the state remains in a similar, essentially classical, state.
\begin{proposition}\label{fin_diag_prop}
Under the assumptions of proposition \ref{key_proposition}, with $\mathcal{H}_{mkt}=\mathbb{C}^N$. Let the initial reduced density matrix be given by equation \ref{rho_0_class_fin}. Finally, we apply the Born-Markov approximation as in proposition \ref{BMA}. Then the market state at time $t$: $\rho_{mkt}(t)$ remains diagonal. In other words we have:
\begin{align*}
\rho_{mkt}(t)&=\sum_{i=1}^Np_i(t)|e_i\rangle\langle e_i|
\end{align*}
where the $\rho_{mkt}(t)$ evolves according to:
\begin{align}\label{G_pde}
\frac{d\rho_{mkt}(t)}{dt}&=\kappa\sum_{i=2}^{N-1}\big(p_{i+1}(t)+p_{i-1}(t)-2p_i(t)\big)|e_i\rangle\langle e_i|+\kappa p_1\big(|e_2\rangle\langle e_2|-|e_1\rangle\langle e_1|\big)\nonumber\\
&+\kappa p_{N}(t)\big(|e_{N-1}\rangle\langle e_{N-1}|-|e_{N}\rangle\langle e_{N}|\big)
\end{align}
\end{proposition}
\begin{proof}
We have:
\begin{align*}
A_u\rho_{mkt}(t)A_d&=\sum_{i=1}^{N-1}p_i(t)|e_{i+1}\rangle\langle e_{i+1}|\\
A_d\rho_{mkt}(t)A_u&=\sum_{i=1}^{N-1}p_i(t)|e_{i}\rangle\langle e_{i}|\\
\frac{1}{2}\{(A_uA_d+A_dA_u),\rho_{mkt}(t)\}&=2\sum_{i=2}^{N-1}p_i(t)|e_i\rangle\langle e_i|+p_1(t)|e_1\rangle\langle e_1|-p_N(t)|e_N\rangle\langle e_N|
\end{align*}
Inserting this into equation \ref{Markov Res} gives the result.
\end{proof}
Note that we can write equation \ref{G_pde} in a standard ``classical'' form as follows:
\begin{align*}
\frac{d\overline{p}}{dt}&=\Delta\overline{p}\text{, }\overline{p}(t)=\{p_1(t),p_2(t),\dots,p_N(t)\}
\end{align*}
Where $\Delta$ represents a discretization of the second derivative operator: $\partial^2/\partial x^2$. Thus, we can see that in the case whereby the market environment is in a thermal equilibrium state, and the market starts in a ``classical'' state, then the market remains in such a state. That is, the state is diagonalized relative to the basis imposed by the eigenstates of the observable we are interested in (the traded price operator). In section \ref{NG}, we consider 2 ways in which we can extend this simple approach:
\begin{itemize}
\item First, in section \ref{NGI}, by allowing non-zero off diagonal terms in the environment state: \ref{identity}.
\item Second, in section \ref{NGII}, by assuming different values for the operators $A_u$ and $A_d$.
\end{itemize}
\section{Non-Gaussian Extensions:}\label{NG}
\subsection{Non-Gaussian Extension I: Non-Commutative State}\label{NGI}
In sections \ref{Markov} and \ref{Gaussian} we have assumed that the environment state remains in the stationary thermal equilibrium state:
\begin{align*}
\rho_{env}&=\frac{1}{K}\sum_{i=1}^K|e_i\rangle\langle e_i|\otimes\rho_B
\end{align*}
Instead, in this section we look at the more general case, and consider the case for non-diagonal $\rho_{env}$:
\begin{align}\label{more_gen_state}
\rho_{env}&=\sum_{l,m=1}^Kr_{lm}|e_l\rangle\langle e_m|\otimes\rho_B
\end{align}
\begin{proposition}\label{NC_state_prop}
We let the environment Hilbert space be given by \ref{H_env}, the full system Hamiltonian by: \ref{Hamiltonian}, with $B_u$/$B_d$ given by \ref{B_def}, and where $H_{env}$ has the product form \ref{product}.

After applying the strong coupling limit \ref{SCL} with the environment state given by \ref{more_gen_state}, we have:
\begin{align}\label{LME_NG2}
\frac{d\rho_{mkt}(t)}{dt}&=-Tr_{env}[H_I(t),\rho^I(0)]\\
&+\sigma^2\Big(A_u\rho_{mkt}(t)A_d+A_d\rho_{mkt}(t)A_u-\frac{1}{2}\{A_uA_d+A_dA_u,\rho_{mkt}(t)\}\Big)\nonumber\\
&+\nu_u^2\Big(A_u\rho_{mkt}(t)A_u-\frac{1}{2}\{A_uA_u,\rho_{mkt}(t)\}\Big)\nonumber\\
&+\nu_d^2\Big(A_d\rho_{mkt}(t)A_d-\frac{1}{2}\{A_dA_d,\rho_{mkt}(t)\}\Big)\nonumber
\end{align}
Where we denote:
\begin{align*}
\sigma^2 &=\kappa\sum_{l=1}^{K-1}r_{ll}\text{, }\nu_u^2=2\kappa\sum_{l=1}^{K-2}r_{l(l+2)}\text{, }\nu_d^2=2\kappa\sum_{l=1}^{K-2}r_{(l+2)l}
\end{align*}
\end{proposition}
\begin{proof}
The result follows from applying the strong coupling limit, and feeding equations \ref{uu}, \ref{dd}, and \ref{ud,du} into proposition \ref{key_proposition}.
\end{proof}
The next proposition shows that the operators in equation \ref{LME_NG2} act diagonally on $\rho_{mkt}$, in the sense that each matrix element in row $i$/column $j$ interacts the element  in row $i+1$/column $j+1$, and row $i-1$/column $j-1$.
\begin{proposition}\label{NG_form}
Let the operators $A_u/A_d$ be given by \ref{A_finite}. Furthermore, we assume $N\rightarrow\infty$ in $\mathbb{C}^N$, so that boundary conditions can be ignored. Finally, we assume that $-Tr_{env}[H_I(t),\rho^I(0)]=0$. Then equation \ref{LME_NG2} can be written:
\begin{align*}
\frac{d\rho_{mkt}(t)}{dt}&=\sigma^2\mathcal{L}\big(\rho_{mkt}(t)\big)-\nu_u\mathcal{L}\big(A_u\rho_{mkt}(t)A_u\big)-\nu_d\mathcal{L}\big(A_d\rho_{mkt}(t)A_d\big)
\end{align*}
Where for infinite dimensional matrix $M$ acting on $\mathbb{C}^{\infty}$, we have:
\begin{align*}
\mathcal{L}(M)_{ij}&=M_{(i+1)(j+1)}+M_{(i-1)(j-1)}-2M_{ij}
\end{align*}
\end{proposition}
\begin{proof}
See appendix, section \ref{appendix}.
\end{proof}
\begin{proposition}\label{derive_var_detail2}
Let the price operator $X$ be given by equation \ref{X_comm}, assume $A_u$/$A_d$ are given by \ref{A_finite}, and let the initial market state (acting on the market Hilbert space): $\rho_0$ be given by:
\begin{align}\label{rho_zero}
\rho_0&=\sum_{i,j}a_{ij}|e_i\rangle\langle e_j|\text{, }\sum_{i=1}^N|a_{ii}|^2=1
\end{align}
then under the time evolution given by proposition \ref{NC_state_prop}, the rate of change in the total variance is given by:
\begin{align}
\frac{\partial(E^{\rho_0}[X^2])}{\partial t}&=\sigma^2\sum_{i=2}^{N-1}x_i^2(a_{(i+1)(i+1)}+a_{(i-1)(i-1)}-2a_{ii})\\
&+\nu_u^2\sum_{i=3}^{N-2}x_i^2\Big(a_{(i-1)(i+1)}-\frac{1}{2}\big(a_{(i-2)i}+a_{i(i+2)}\big)\Big)\nonumber\\
&+\nu_d^2\sum_{i=3}^{N-2}x_i^2\Big(a_{(i+1)(i-1)}-\frac{1}{2}\big(a_{i(i-2)}+a_{(i+2)i}\big)\Big)\nonumber
\end{align}
where $\nu_u$ and $\nu_d$ are determined by the off-diagonal terms of the environment state (per proposition \ref{NC_state_prop}).
\end{proposition}
\begin{proof}
See appendix, section \ref{appendix}.
\end{proof}
Note that due to the presence of off-diagonal terms in the environment density matrix (parameterized in this case by $\nu_u$ and $\nu_d$), we have both Gaussian contributions to the evolution of the variance:
\begin{align*}
\sigma^2\sum_{i=2}^{N-1}x_i^2(r_{(i+1)(i+1)}+r_{(i-1)(i-1)}-2r_{ii})
\end{align*}
and also non-Gaussian terms that depend on the off-diagonal components:
\begin{align*}
&+\nu_u^2\sum_{i=3}^{N-2}x_i^2\Big(a_{(i-1)(i+1)}-\frac{1}{2}\big(a_{(i-2)i}+a_{i(i+2)}\big)\Big)\\
&+\nu_d^2\sum_{i=3}^{N-2}x_i^2\Big(a_{(i+1)(i-1)}-\frac{1}{2}\big(a_{i(i-2)}+a_{(i+2)i}\big)\Big)
\end{align*}
If the terms $\nu_u$, $\nu_d$ are zero, then the evolution will be Gaussian, and a diagonal market density matrix will remain diagonal. If $\nu_u\neq 0$ or $\nu_d\neq 0$, then even if the density matrix starts in a diagonal (classical) state, the non-Gaussian evolution will evolve non-zero off-diagonal terms as the simulation progresses.
\subsection{Non-Gaussian Extension II: Non-Local Operators}\label{NGII}
In section \ref{Gaussian}, we have studied the case whereby the market response to a change in the environment to a higher level of risk appetite, is that the price jumps by a fixed amount. Ie, if $\rho_{mkt}(0)=|e_i\rangle\langle e_i|$, then the initial price is given by the eigenvalue: $x_i$, and the market response to an increase in risk appetite would be given by:
\begin{align*}
A_u|e_i\rangle&=|e_{i+1}\rangle\\
X|e_{i+1}\rangle&=x_{i+1}|e_{i+1}\rangle
\end{align*}
So that the price jumps from $x_i$ to $x_{i+1}$. Now, consider the case whereby the response to an increase in the environment risk appetite is uncertain.  The price may jump by 1 level, or more, or the price may not jump at all. In order to introduce this effect, we apply a discrete convolution with a probability distribution labelled $P_H$:
\begin{align}\label{P_H_def}
P_H&=\sum_ih_i|e_i\rangle
\end{align}
We start with the operator defined by equation \ref{A_finite}:
\begin{align*}
A_u|e_i\rangle&=|e_{i+1}\rangle\text{, }i<N
\end{align*}
In order to apply the convolution between $|\psi\rangle$ and $H$, we can use the following operator:
\begin{definition}\label{finite_H_def}
\begin{align*}
H&=\sum_{j=1}^N\sum_{k=1-j}^{N-j}h_k|e_{j+k}\rangle\langle e_j| 
\end{align*}
The operators $A^H_u$, $A^H_d$, together with the new interaction Hamiltonian: $H_I$ can now be defined by:
\begin{align*}
A^H_u&=A_uH\\
&=\sum_{i=1}^{N-1}\sum_{j=1}^Nh_{i-j}|e_{i+1}\rangle\langle e_j|\\
A^H_d&=A^{H\dagger}_u\\
&=H^{\dagger}A_d
\end{align*}
\end{definition}
Using this definition for the operators $A_u^H$, and $A_d^H$, we can now restate proposition: \ref{fin_diag_prop}, with definition \ref{finite_H_def} in place of definition \ref{A_finite}.
\begin{proposition}\label{LME_quantum_fin}
Let the market Hilbert space be given by: $\mathcal{H}_{mkt}=\mathbb{C}^N$. Furthermore, assume definition \ref{finite_H_def} applies with regards to the interaction Hamiltonian. Then, after applying the Born-Markov approximation given in proposition \ref{BMA}, the reduced density matrix for the market evolves according to:
\begin{align}\label{C_rho_eqn_H_fin}
\frac{d\rho_{mkt}(t)}{dt}&=\sigma^2\Big(A_u^H\rho_{mkt}(t)A_d^H+A_d^H\rho_{mkt}(t)A_u^H-\frac{1}{2}\big\{A_u^HA_d^H+A_d^HA_u^H,\rho_{mkt}(t)\big\}\Big)
\end{align}
\end{proposition}
\begin{proof}
The result follows from inserting the operators given in definition \ref{finite_H_def} into proposition \ref{BMA}.
\end{proof}
In order to gain a qualitative understanding of the dynamics described by proposition \ref{LME_quantum_fin}, we consider the simplified setup whereby:
\begin{align}\label{rev_NG}
P_H&=(h_1,h_0,h_1)\\
h_0,h_1&\in\mathbb{R}\nonumber\\
h_0^2+2h_1^2&=1\nonumber
\end{align}
We set the initial state to
\begin{align}\label{rev_NG_state}
\rho_0&=\sum_{i,j=1}^Na_{ij}|e_i\rangle\langle e_j|\\
\sum_i|a_{ii}|^2&=1\nonumber\\
a_{ij}&\rightarrow 0\text{ for }i,j\rightarrow 1,N\nonumber\\
N&\rightarrow\infty\nonumber
\end{align}
The purpose of proposition \ref{derive_var_detail}, is to highlight why (in this simplified case) the variance depends not just on the diagonal elements of the density matrix, but also on the non-diagonal elements. In other words, we show why in the quantum case, 2 density matrices with the same initial probability distribution for $X$, have different variance growth rates. 
\begin{proposition}\label{derive_var_detail}
Let the price operator $X$ be given by equation \ref{X_comm}, the vector $P_H$ by equation \ref{rev_NG}, and the initial state: $\rho_0$, by equation \ref{rev_NG_state}, then under the time evolution given by proposition \ref{LME_quantum_fin}, the rate of change in the total variance is given by:
\begin{align}
\frac{\partial(E^{\rho_0}[X^2])}{\partial t}&=\sigma^2\sum_{i=2}^{N-1}x_i^2h_0^2(a_{(i+1)(i+1)}+a_{(i-1)(i-1)}-2a_{ii})\\
&+\sum_{i=3}^{N-2}x_i^2h_1^2(a_{(i+2)(i+2)}+a_{(i-2)(i-2)}-2a_{ii})\Big)\nonumber\\
&+\sigma^2\sum_{i=3}^{N-2}h_0h_1x_i^2\Big(a_{(i+1)(i+2)}+a_{(i+2)(i+1)}+a_{(i-1)(i-2)}+a_{(i-2)(i-1)}\nonumber\\
&-a_{i(i+1)}+a_{(i+1)i}+a_{i(i-1)}+a_{(i-1)i}\Big)\nonumber
\end{align}
\end{proposition}
\begin{proof}
See appendix \ref{appendix}.
\end{proof}
Note that in proposition \ref{derive_var_detail} we have:
\begin{itemize}
\item Standard Gaussian terms: $h_0^2(a_{(i+1)(i+1)}+a_{(i-1)(i-1)}-2a_{ii})$
\item Non-Gaussian terms that act on the diagonal of the density matrix: 

$h_1^2(a_{i+2}^2+a_{i-2}^2-2a_i^2)$
\item Non-Gaussian terms that act on the off-diagonal elements of the density matrix: $\sigma^2h_0h_1x_i^2(a_{(i+1)(i+2)}+a_{(i+2)(i+1)}+a_{(i-1)(i-2)}+a_{(i-2)(i-1)}-a_{i(i+1)}+a_{(i+1)i}+a_{i(i-1)}+a_{(i-1)i})$
\end{itemize}
\subsection{Numerical Illustrations:}
In this section, we seek to illustrate the impact that the non-Gaussian extension discussed in section \ref{NGI}, and the level of the market entropy, have on the evolution of the resulting probability distribution. To start with, in section \ref{setup}, we describe the setup for some numerical simulations, before discussing the results in section \ref{results}.

Note that in this section the focus is on the statistical properties of the resulting random processes, rather than any specific application to financial underlyings.
\subsubsection{Basic Setup:}\label{setup}
The market Hilbert space is set to $\mathcal{H}_{mkt}=\mathbb{C}^{1001}$. In this case we set the underlying random variable, $X$ to:
\begin{align*}
X&=\sum_{i=1}^{1001}x_i|e_i\rangle\langle e_i|\\
x_i&=-\frac{1}{2}+\frac{i-1}{1000}
\end{align*}
We define the initial market state as follows, where $N(x,\mu,\sigma)$ is the normal distribution density function with mean $\mu$, and variance $\sigma^2$:
\begin{align}\label{define_theta}
\rho_0(\theta)&=\theta\rho_c+(1-\theta)\rho_q\\
\rho_c&=\sum_{i=1}^{1001}p_i|e_i\rangle\langle e_i|\text{, }p_i=N(x_i,0,0.005)\nonumber\\
\rho_q&=(\overline{P}\cdot\overline{P}^T)\text{, }\overline{P}^T=(\sqrt{p_1}\text{, }\sqrt{p_2}\text{, }\dots\text{, }\sqrt{p_{1001}})\nonumber
\end{align}
This ensures that $\rho_0(0)$ is a pure state, and $\rho_0(1)$ is a maximum entropy diagonal state, and also that the initial probability distribution for $X$ is unaffected by the choice of $\theta$. Figure \ref{mkt_ent} shows the initial entropy as a function of $\theta$ for the initial market state.
\begin{figure}
\includegraphics[scale=0.75]{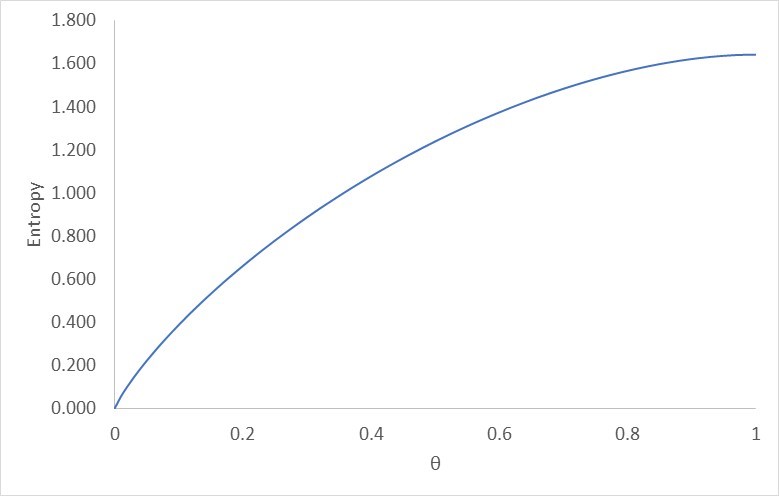}
\caption{Entropy for the initial market state, as a function of $\theta$.}\label{mkt_ent}
\end{figure}
We set the environment dimension to $K=11$. Finally, for $\sigma^2$, we set:
\begin{align*}
\sigma^2&=\bigg(\frac{0.02}{\delta x}\bigg)^2\\
&=400
\end{align*}
Under these conditions, we run the following:
\begin{enumerate}
\item[{\em Sim 1}] Gaussian simulation, based on proposition \ref{BMA} and definition \ref{A_finite}.
\item[{\em Sim 2}] Non-Gaussian simulation based on proposition \ref{NG_form}. We scale $\nu_u/\nu_d$ in steps from zero up to $\nu_u/\nu_d=\sigma^2$.
\item[{\em Sim 3}] Non-Gaussian simulation based on proposition \ref{LME_quantum_fin} with $P_H$ given by equation \ref{rev_NG}. We set: $h_0^2=1-2h$, $h_1^2=h$, with $h$ varying in steps from $0$ to $20\%$.
\end{enumerate}
\subsubsection{Sim 1:}\label{results}
Figure \ref{res_1} shows the variance \& the entropy gain for the Gaussian simulation. The variance does not depend on the entropy of the initial state. This is because under proposition \ref{BMA}, with definition \ref{A_finite}, there is no interaction between the diagonal of $\rho_{mkt}$ and the off diagonal elements. A classical density matrix, remains classical (ie diagonalized relative to the main traded price operator).

The entropy gain over the simulation is greater where there is lower entropy to start with. Ie more information regarding the market is lost, if there is more to lose to begin with.
\begin{figure}
\includegraphics[scale=0.7]{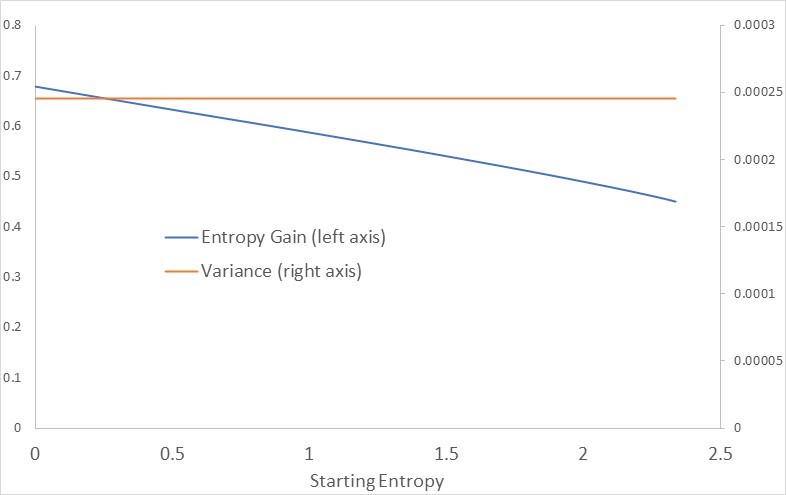}
\caption{The chart shows the impact of the starting entropy on the Gaussian evolution, under the simulation described above: {\em Sim 1}. The horizontal axis shows the value of $\theta$ in equation \ref{define_theta}. The left axis shows the gain in Von-Neumann entropy, and the right axis shows the variance after 1000 time-steps.}\label{res_1}
\end{figure}
\subsubsection{Sim 2:}
Figure \ref{res_2} shows the results from the simulation based on Non-Gaussian Extension I, with a classical initial state. Note first that the variance is not dependent on the choice of $\nu_u/\nu_d$. Both the probability distributions shown have the same variance. As described in section \ref{NGI}, the non-Gaussian contributions to the first time-step depend on the off diagonal elements, which are all zero initially. Then since variance is linear, the variance introduced in future time-steps must match the first, and thus do not depend on $\nu_u/\nu_d$.

The same does not apply for the higher moments of the distribution, which do depend on $\nu_u/\nu_d$. Even where the state starts classical, it does not remain so. The non-Gaussian elements drive excess kurtosis, which in turn effects the entropy gain or information loss over the simulation. A higher rate of retention of information (lower entropy gain), is associated with higher excess kurtosis. Note that we define excess kurtosis as:
\begin{align*}
\text{Excess Kurtosis}&=\frac{E[X^4]-(3*E[X^2])^2}{(3*E[X^2])^2}
\end{align*}
Figure \ref{res_3} shows the variance of the process as a function of $\nu_u/\nu_d$, for a zero entropy state, and a classical entropy state. For the zero entropy initial state, the off-diagonal elements in $\rho_{mkt}(0)$ are non-zero from the beginning, which means that changing $\nu_u/\nu_d$ impacts the variance where the initial state is non-classical.
\begin{figure}
\includegraphics[scale=0.4]{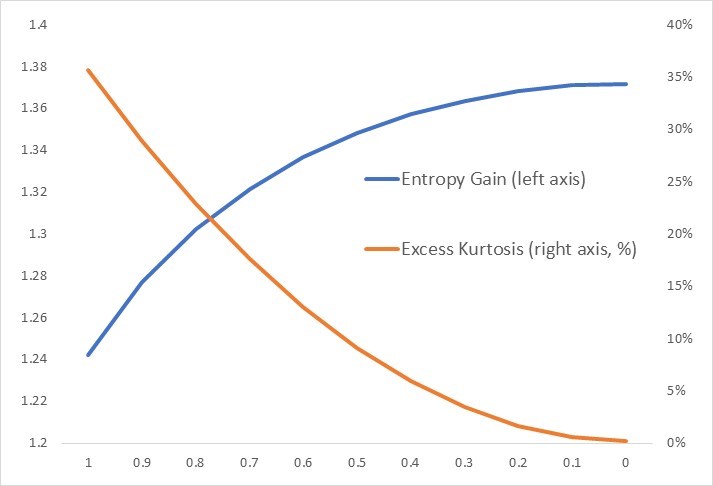}
\includegraphics[scale=0.4]{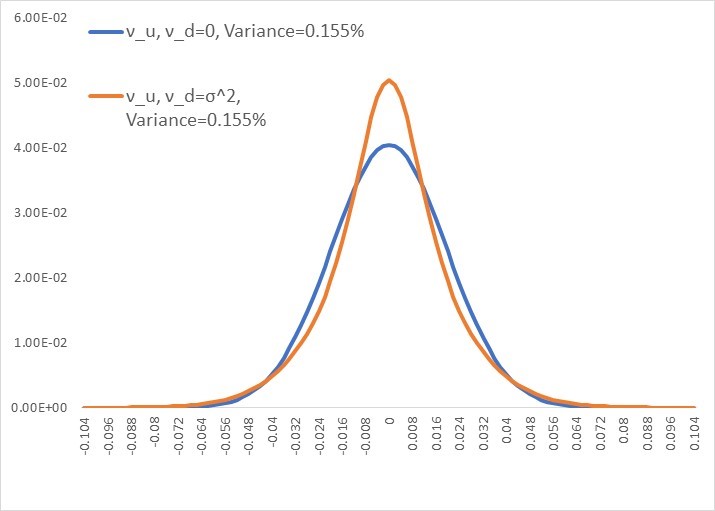}
\caption{The left chart shows the excess Kurtosis, and the entropy gained, after 1000 time-steps of the Non-Gaussian simulation, as described above: {\em Sim 2}, with a classical initial state. The horizontal axis shows the ratio of $\nu_u/\nu_d$ as a fraction of $\sigma^2$. The right chart shows the final distributions for the price observable for $\nu_u/\nu_d=0$ and $\nu_u/\nu_d=\sigma^2$.}\label{res_2}
\end{figure}
\begin{figure}
\includegraphics[scale=0.7]{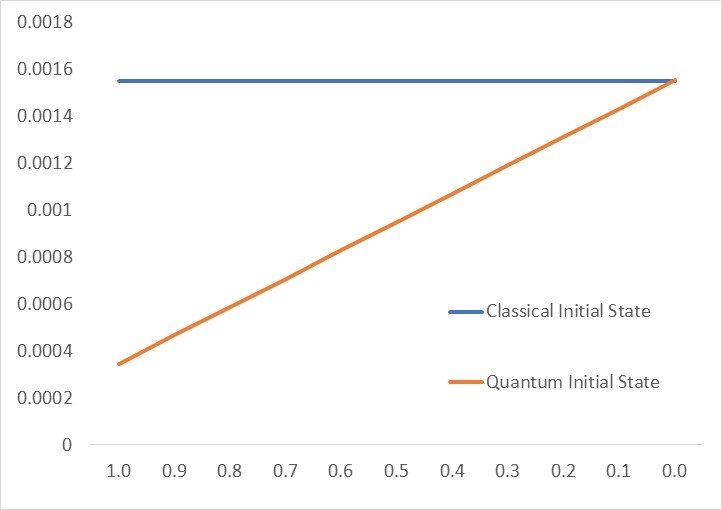}
\caption{The chart shows the variance after 1000 time-steps, for the simulation described above: {\em Sim 2}. The horizontal axis shows the ratio of $\nu_u/\nu_d$ to $\sigma^2$. The blue line shows results with a classical initial state, and the orange line shows the results for a zero entropy initial state.}\label{res_3}
\end{figure}
\subsubsection{Sim 3:}
Figure \ref{res_4} shows the results for the simulation based on the Non-Gaussian extension II, with a classical initial state. Here the increase in excess kurtosis are non-monotonic as the contribution from the non-Gaussian components are scaled up. This is because now increasing the value of $h$ increases both the fourth moment $E[X^4]$ and the variance. The simulation still confirms the inverse relationship between the entropy gain/information loss over the simulation, and the excess kurtosis. 
\begin{figure}
\includegraphics[scale=0.7]{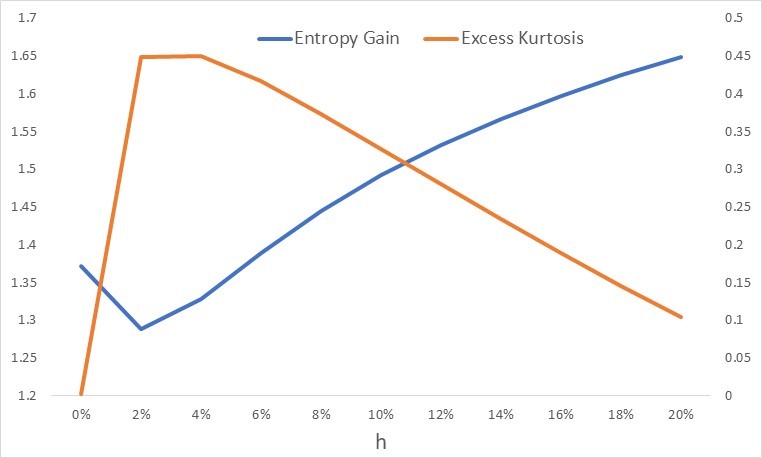}
\caption{The chart shows the excess Kurtosis, and the entropy gained, after 1000 time-steps of the Non-Gaussian simulation, as described above: {\em Sim 3}. The horizontal axis shows the chosen value for $h$ which is increased from $0\%$ (Gaussian case) to a maximum of $20\%$.}\label{res_4}
\end{figure}
\section{Conclusion:}
Fundamentally, whenever one turns to probability in financial modelling, this is motivated by a lack of information, or knowledge regarding the future. We have argued in this article that the gradual loss of information regarding financial prices, is a key statistic to monitor, alongside factors such as the variance. Furthermore, we have illustrated, in section \ref{entropy}, hypothetical examples of where there can be differing levels of information regarding the state of the market. Further consideration of this concept, for example the relation between market micro-structure and market entropy, are an interesting avenue for future research.

The open quantum systems approach, discussed in this article, represents a more flexible means of modelling the random evolution of the financial market, and in particular, modelling the degree of entropy gain as we model further and further into the future.

In the article, the results we have presented depend largely on numerical simulation, based on discretized models. Developing analytic, or semi-analytic, solutions to the non-Gaussian approaches is another key avenue for future research.
\section{Appendix: Detailed Derivations}\label{appendix}
\begin{proof}[Proof of Proposition \ref{entropy_proposition}]
First of all, we consider a state $\rho\in\mathcal{A}$, and write out the spectral resolution in some orthonormal basis: $|\phi_i\rangle$:
\begin{align*}
\rho&=\sum_{j=1}^Nq_j|\phi_j\rangle\langle\phi_j|
\end{align*}
Where, $q_j\geq 0$ (some $q_j$ could be zero). Then by assumption:
\begin{align*}
p_i&=Tr[\rho P_i]\\
&=\sum_{j=1}^Nq_j|\langle\phi_j|e_i\rangle|^2
\end{align*}
Then we have:
\begin{align*}
S(\rho_{classical})&=-\sum_{i=1}^Np_i\log(p_i)\\
&=-\sum_{i=1}^N\Big(\sum_{j=1}^Nq_j|\langle e_i|\phi_j\rangle|^2\Big)\log\Big(\sum_{j=1}^Nq_j|\langle e_i|\phi_j\rangle|^2\Big)
\end{align*}
We now label: $a_{ij}=|\langle e_i|\phi_j\rangle|^2$, and note that:
\begin{align*}
a_{ij}&\geq 0\\
\sum_{i=1}^Na_{ij}&=\sum_{j=1}^Na_{ij}=1
\end{align*}
Therefore, we have:
\begin{align*}
S(\rho_{classical})&=-\sum_{i=1}^N\Big(\sum_{j=1}^Nq_ja_{ij}\Big)\log\Big(\sum_{j=1}^Nq_ja_{ij}\Big)\\
&=\sum_{i=1}^Nf\Big(\sum_{j=1}^Nq_ja_{ij}\Big)\text{, where }f(x)=-x\log(x)\\
&=\sum_{i=1}^Nf(q_1a_{i1}+q_2a_{i2}+...+q_Na_{iN})
\end{align*}
We have: $f''(x)\leq 0\text{, for }x\geq 0$. So therefore, $f(x)$ is a concave function, and we have from Jensen's inequality that:
\begin{align*}
S(\rho_{classical})&\geq \sum_{i=1}^N a_{i1}f(q_1)+a_{i2}f(q_2)+...+a_{iN}f(q_N)\\
&=\Big(\sum_{i=1}^Na_{i1}\Big)f(q_1)+\Big(\sum_{i=1}^Na_{i2}\Big)f(q_2)+...+\Big(\sum_{i=1}^Na_{iN}\Big)f(q_N)\\
&=\sum_{j=1}^Nf(q_j)\\
&=S(\rho)
\end{align*}
\end{proof}
\begin{proof}[Proof of Proposition \ref{NG_form}]
We start with the Lindblad master equation, and consider each term:
\begin{align*}
\frac{d\rho_{mkt}(t)}{dt}&=-Tr_{env}[H_I(t),\rho^I(0)]\\
&+\sigma^2\Big(A_u\rho_{mkt}(t)A_d+A_d\rho_{mkt}(t)A_u-\frac{1}{2}\{A_uA_d+A_dA_u,\rho_{mkt}(t)\}\Big)\\
&+\nu_u^2\Big(A_u\rho_{mkt}(t)A_u-\frac{1}{2}\{A_uA_u,\rho_{mkt}(t)\}\Big)\\
&+\nu_d^2\Big(A_d\rho_{mkt}(t)A_d-\frac{1}{2}\{A_dA_d,\rho_{mkt}(t)\}\Big)
\end{align*}
We find that:
\begin{align*}
(A_uA_d+A_dA_u)\sum_{i,j=-\infty}^{\infty}M_{ij}|e_i\rangle\langle e_j|&=2\sum_{i,j=-\infty}^{\infty}M_{ij}|e_i\rangle\langle e_j|\\
A_u\sum_{i,j=-\infty}^{\infty}M_{ij}|e_i\rangle\langle e_j|A_d&=\sum_{i,j=-\infty}^{\infty}M_{(i-1)(j-1)}|e_i\rangle\langle e_j|\\
A_d\sum_{i,j=-\infty}^{\infty}M_{ij}|e_i\rangle\langle e_j|A_u&=\sum_{i,j=-\infty}^{\infty}M_{(i+1)(j+1)}|e_i\rangle\langle e_j|
\end{align*}
Similarly, we have that:
\begin{align*}
A_uA_u\sum_{i,j=-\infty}^{\infty}M_{ij}|e_i\rangle\langle e_j|&=A_u\Big(A_u\sum_{i,j=-\infty}^{\infty}M_{ij}|e_i\rangle\langle e_j|A_u\Big)A_d
\end{align*}
Similarly:
\begin{align*}
\sum_{i,j=-\infty}^{\infty}M_{ij}|e_i\rangle\langle e_j|A_uA_u&=A_d\Big(A_u\sum_{i,j=-\infty}^{\infty}M_{ij}|e_i\rangle\langle e_j|A_u\Big)A_u\\
A_dA_d\sum_{i,j=-\infty}^{\infty}M_{ij}|e_i\rangle\langle e_j|&=A_d\Big(A_d\sum_{i,j=-\infty}^{\infty}M_{ij}|e_i\rangle\langle e_j|A_d\Big)A_u\\
\sum_{i,j=-\infty}^{\infty}M_{ij}|e_i\rangle\langle e_j|A_dA_d&=A_u\Big(A_d\sum_{i,j=-\infty}^{\infty}M_{ij}|e_i\rangle\langle e_j|A_d\Big)A_d
\end{align*}
The result follows from feeding this into the Lindblad master equation.
\end{proof}
\begin{proof}[Proof of Proposition \ref{derive_var_detail2}]
First note that:
\begin{align}\label{star}
\frac{\partial(E^{\rho_0}[X^2])}{\partial t}&=\frac{\partial(Tr[X^2\rho_0(t)])}{\partial t}\nonumber\\
&=\frac{\partial}{\partial t}\bigg(\sum_{i=1}^Nx_i^2a_{ii}(t)\bigg)\nonumber\\
&=\sum_{i=1}^Nx_i^2\frac{\partial a_{ii}(t)}{\partial t}\nonumber\\
&=Tr\Big[X^2\frac{\partial\rho_0(t)}{\partial t}\Big]
\end{align}
Where the third line of \ref{star} follows from the fact that under the system Hamiltonian \ref{H_sys} we have:
\begin{align*}
e^{-iH_{sys}t}Xe^{iH_{sys}t}&=X
\end{align*}
Applying proposition \ref{NC_state_prop} to equation \ref{star}, we get:
\begin{align}\label{star2}
\frac{\partial(E^{\rho_0}[X^2])}{\partial t}&=\sigma^2Tr\Big[X^2\Big(A_u\rho_{mkt}(t)A_d+A_d\rho_{mkt}(t)A_u\\
&-\frac{1}{2}\{A_uA_d+A_dA_u,\rho_{mkt}(t)\}\Big)\Big]\nonumber\\
+\nu_u^2&Tr\Big[X^2\Big(A_u\rho_{mkt}(t)A_u-\frac{1}{2}\{A_uA_u,\rho_{mkt}(t)\}\Big)\Big]\nonumber\\
+\nu_d^2&Tr\Big[X^2\Big(A_d\rho_{mkt}(t)A_d-\frac{1}{2}\{A_dA_d,\rho_{mkt}(t)\}\Big)\Big]\nonumber
\end{align}
Under definition \ref{A_finite} and equation \ref{rho_zero}, we have:
\begin{align*}
Tr&\Big[X^2\Big(A_u\rho_{mkt}(t)A_d+A_d\rho_{mkt}(t)A_u-\frac{1}{2}\{A_uA_d+A_dA_u,\rho_{mkt}(t)\}\Big)\Big]\\
&=\sum_{i=2}^{N-1}x_i^2(a_{(i+1)(i+1)}+a_{(i-1)(i-1)}-2a_{ii})\\
Tr&\Big[X^2\Big(A_u\rho_{mkt}(t)A_u-\frac{1}{2}\{A_uA_u,\rho_{mkt}(t)\}\Big)\Big]\\
&=\sum_{i=3}^{N-2}x_i^2\Big(a_{(i-1)(i+1)}-\frac{1}{2}\big(a_{(i-2)i}+a_{i(i+2)}\big)\Big)\\
Tr&\Big[X^2\Big(A_d\rho_{mkt}(t)A_d-\frac{1}{2}\{A_dA_d,\rho_{mkt}(t)\}\Big)\Big]\\
&=\sum_{i=3}^{N-2}x_i^2\Big(a_{(i+1)(i-1)}-\frac{1}{2}\big(a_{i(i-2)}+a_{(i+2)i}\big)\Big)
\end{align*}
The result follows by feeding this into equation \ref{LME_NG2}.
\end{proof}
\begin{proof}[Proof of Proposition \ref{derive_var_detail}]
From proposition \ref{derive_var_detail2} we have that:
\begin{align}\label{star3}
\frac{\partial(E^{\rho_0}[X^2])}{\partial t}&=Tr\Big[X^2\frac{\partial\rho_0(t)}{\partial t}\Big]
\end{align}
Applying proposition \ref{LME_quantum_fin} to equation \ref{star3}, we get:
\begin{align}\label{star4}
\frac{\partial(E^{\rho_0}[X^2])}{\partial t}&=\sigma^2Tr\Big[X^2(A_u^H\rho_0A_d^H+A_d^H\rho_0A_u^H-\frac{1}{2}\{A_u^HA_d^H+A_d^HA_u^H,\rho_0\}\Big]
\end{align}
We now calculate the terms in \ref{star3} using definition \ref{finite_H_def}, and proposition \ref{LME_quantum_fin}.
Under \ref{rev_NG} we have:
\begin{align}\label{app_AuH}
A_d^H&=\sum_{i=1}^{N-1}\big(h_0|e_i\rangle\langle e_{i+1}|+h_{-1}|e_{i+1}\rangle\langle e_{i+1}|\big)+\sum_{i=2}^{N-1}h_1|e_{i-1}\rangle\langle e_{i+1}|\\
A_u^H&=\sum_{i=1}^{N-1}\big(h_0|e_{i+1}\rangle\langle e_i|+h_{-1}|e_{i+1}\rangle\langle e_{i+1}|\big)+\sum_{i=2}^{N-1}h_1|e_{i+1}\rangle\langle e_{i-1}|\nonumber
\end{align}
Since $\rho_0=\sum_{i,j=1}^Na_{ij}|e_i\rangle\langle e_j|$, first we consider an individual $A_u^Ha_{ij}|e_i\rangle\langle e_j|A_d^H$ term. We get:
\begin{align}\label{app_au_rho_ad}
A_u^Ha_{ij}|e_i\rangle\langle e_j|A_d^H&=a_{ij}\Big(h_0^2|e_{i+1}\rangle\langle e_{j+1}|+h_1^2|e_{i+2}\rangle\langle e_{j+2}|+h_{-1}^2|e_i\rangle\langle e_i|\\
&+h_0h_1(|e_{i+2}\rangle\langle e_{j+1}|+|e_{i+1}\rangle\langle e_{j+2}|)\nonumber\\
&+h_0h_{-1}(|e_i\rangle\langle e_{j+1}|+|e_{i+1}\rangle\langle e_j|)\nonumber\\
&+h_{-1}h_1(|e_i\rangle\langle e_{j+2}|+|e_{i+2}\rangle\langle e_j|)\Big)\nonumber
\end{align}
Collecting the diagonal terms from \ref{app_au_rho_ad}, we get:
\begin{align}\label{app_OQS_diag_1}
A_u^H\rho_0A_d^H&=\sum_{i=1}^{N-2}a_{ii}h_0^2|e_{i+1}\rangle\langle e_{i+1}|+\sum_{i=1}^{N-2}a_{ii}h_1^2|e_{i+2}\rangle\langle e_{i+2}|+\sum_{i=1}^Na_{ii}h_{-1}^2|e_i\rangle\langle e_i|\\
&+\sum_{i=1}^{N-2}(a_{i(i+1)}+a_{(i+1)i})h_0h_1|e_{i+2}\rangle\langle e_{i+2}|\nonumber\\
&+\sum_{i=1}^{N-1}(a_{i(i+1)}+a_{(i+1)i})h_0h_{-1}|e_{i+1}\rangle\langle e_{i+1}|\nonumber\\
&+\sum_{i=1}^{N-2}(a_{i(i+2)}+a_{(i+2)i})h_{-1}h_1|e_{i+2}\rangle\langle e_{i+2}|\nonumber
\end{align}
Similarly, collecting together the diagonal terms from $A_d^Ha_{ij}|e_i\rangle\langle e_j|A_u^H$, we get:
\begin{align}\label{app_OQS_diag_2}
A_d^H\rho_0A_u^H&=\sum_{i=2}^Na_{ii}h_0^2|e_{i-1}\rangle\langle e_{i-1}|+\sum_{i=3}^Na_{ii}h_1^2|e_{i-2}\rangle\langle e_{i-2}|+\sum_{i=1}^Na_{ii}h_{-1}^2|e_i\rangle\langle e_i|\\
&+\sum_{i=2}^{N-1}(a_{i(i+1)}+a_{(i+1)i})h_0h_1|e_{i-1}\rangle\langle e_{i-1}|\nonumber\\
&+\sum_{i=1}^{N-1}(a_{i(i+1)}+a_{(i+1)i})h_0h_{-1}|e_i\rangle\langle e_i|\nonumber\\
&+\sum_{i=1}^{N-2}(a_{i(i+2)}+a_{(i+2)i})h_{-1}h_1|e_{i}\rangle\langle e_{i}|\nonumber
\end{align}
We now consider the individual $A_u^HA_d^Ha_{ij}|e_i\rangle\langle e_j|$, and $A_u^HA_d^Ha_{ij}|e_i\rangle\langle e_j|$ terms.
\begin{align*}
A_u^HA_d^Ha_{ij}|e_i\rangle\langle e_j|&=A_d^HA_u^Ha_{ij}|e_i\rangle\langle e_j|\\
=a_{ij}\Big((h_{-1}^2&+h_0^2+h_1^2)|e_i\rangle\langle e_j|+(h_0h_1+h_{-1}h_0)(|e_{i+1}\rangle\langle e_j|+|e_{i-1}\rangle\langle e_j|)\\
&+h_{-1}h_1(|e_{i+2}\rangle\langle e_j|+|e_{i-2}\rangle\langle e_j|)\Big)
\end{align*}
Again, collecting together the diagonal terms, we get:
\begin{align}\label{app_diag_3}
\frac{1}{2}(A_d^HA_u^H+A_u^HA_d^H)\rho_0&=\sum_{i=1}^Na_{ii}(h_{-1}^2+h_0^2+h_1^2)|e_i\rangle\langle e_i|\\
&+\sum_{i=1}^{N-1}(h_0h_1+h_{-1}h_0)(a_{i(i+1)}|e_{i+1}\rangle\langle e_{i+1}|+a_{(i+1)i}|e_i\rangle\langle e_i|)\nonumber\\
&+\sum_{i=1}^{N-2}h_{-1}h_1(a_{(i+2)i}|e_i\rangle\langle e_i|+a_{i(i+2)}|e_{i+2}\rangle\langle e_{i+2}|)\nonumber
\end{align}
Finally, we consider the individual $a_{ij}|e_i\rangle\langle e_j|A_u^HA_d^H$ terms.
\begin{align*}
a_{ij}|e_i\rangle\langle e_j|A_u^HA_d^H&=a_{ij}|e_i\rangle\langle e_j|A_d^HA_u^H\\
&=a_{ij}\Big((h_{-1}^2+h_0^2+h_1^2)|e_i\rangle\langle e_j|+(h_0h_1+h_{-1}h_0)(|e_i\rangle\langle e_{j+1}|\\
&+|e_i\rangle\langle e_{j-1}|)+h_{-1}h_1(|e_i\rangle\langle e_{j+2}|+|e_i\rangle\langle e_{j-2}|)\Big)
\end{align*}
So that for the diagonal terms we get:
\begin{align}\label{app_diag_4}
\frac{1}{2}\rho_0(A_d^HA_u^H+A_u^HA_d^H)&=\sum_{i=1}^Na_{ii}(h_{-1}^2+h_0^2+h_1^2)|e_i\rangle\langle e_i|\\
&+\sum_{i=1}^{N-1}(h_0h_1+h_{-1}h_0)(a_{(i+1)i}|e_{i+1}\rangle\langle e_{i+1}|+a_{i(i+1)}|e_i\rangle\langle e_i|)\nonumber\\
&+\sum_{i=1}^{N-2}h_{-1}h_1(a_{(i+2)i}|e_{i+2}\rangle\langle e_{i+2}|+a_{i(i+2)}|e_i\rangle\langle e_i|\nonumber
\end{align}
We now feed equations \ref{app_OQS_diag_1}, \ref{app_OQS_diag_2}, \ref{app_diag_3} and \ref{app_diag_4} into equation \ref{star2}. We group the terms together by the coefficients of $h_ih_j$. First note, that the terms in $h_{-1}^2$ cancel to zero, and for $h_0^2$ and $h_1^2$, we get:
\begin{align}\label{app_fin_1}
&\sigma^2\sum_{i=3}^{N-2}h_0^2x_i^2(a_{(i+1)(i+1)}+a_{(i-1)(i-1)}-2a_{ii})\\
&+h_1^2x_i^2(a_{(i+2)(i+2)}+a_{(i-2)(i-2)}-2a_{ii})\nonumber
\end{align}
The terms in $h_{-1}h_0$ and $h_{-1}h_1$ also cancel out, leaving the terms in $h_0h_1$:
\begin{align}\label{app_fin_2}
&\sigma^2\sum_{i=3}^{N-2}h_0h_1x_i^2(a_{(i+1)(i+2)}+a_{(i+2)(i+1)}+a_{(i-1)(i-2)}+a_{(i-2)(i-1)}\\
&-a_{i(i+1)}+a_{(i+1)i}+a_{i(i-1)}+a_{(i-1)i})\nonumber
\end{align}
So that finally we have:
\begin{align}\label{app_final_final}
\frac{\partial(E^{\rho_0}[X^2])}{\partial t}&=\sigma^2\sum_{i=2}^{N-1}x_i^2h_0^2(a_{(i+1)(i+1)}+a_{(i-1)(i-1)}-2a_{ii})\\
&+\sum_{i=3}^{N-2}h_1^2(a_{(i+2)(i+2)}+a_{(i-2)(i-2)}-2a_{ii})\\
&+\sigma^2\sum_{i=3}^{N-2}h_0h_1x_i^2(a_{(i+1)(i+2)}+a_{(i+2)(i+1)}+a_{(i-1)(i-2)}+a_{(i-2)(i-1)}\nonumber\\
&-a_{i(i+1)}+a_{(i+1)i}+a_{i(i-1)}+a_{(i-1)i})\nonumber
\end{align}
\end{proof}
\section*{Acknowledgements}
The author would like to thank Dr Marco Merkli, and Dr Emmanuel Haven for their support and advice.
\section*{Declaration of Conflicting Interests}
The author declares no potential conflicts of interest with respect to the research, authorship, and/or publication of this article.

\end{document}